\documentclass[amssymb, nofootinbib, superscriptaddress]{revtex4}
\usepackage{hyperref}
\usepackage{amsfonts,amssymb,amsmath,exscale,bbm}
\usepackage{footnpag} 
\usepackage[all,knot]{xy}  
\usepackage{color}
\usepackage{graphicx}
\usepackage{epsfig}
\usepackage{subfig}
\usepackage{stmaryrd}

\usepackage{amsmath}
\usepackage{amsfonts}
\usepackage{graphicx}
\usepackage{psfrag}
\usepackage{array}
\usepackage{bbm}
\usepackage{hyperref}
\usepackage{amsthm}
\theoremstyle{definition}

\newtheorem{proposition}{Proposition}
\newtheorem{corollary}{Corollary}

\newcommand{\tr}{\mathrm{tr}}

\newcommand{\cA}{{\mathcal A}}
\newcommand{\cB}{{\mathcal B}}
\newcommand{\cC}{{\mathcal C}}

\newcommand{\cE}{{\mathcal E}}
\newcommand{\cF}{{\mathcal F}}
\newcommand{\cG}{{\mathcal G}}

\newcommand{\cL}{{\mathcal L}}

\newcommand{\cN}{{\mathcal N}}

\newcommand{\cP}{{\mathcal P}}

\newcommand{\cT}{{\mathcal T}}
\newcommand{\cV}{{\mathcal V}}

\newcommand{\cZ}{{\mathcal Z}}

\def\inv{{\mbox{\tiny -1}}}

\newcommand\beq{\begin{equation}}
\newcommand\eeq{\end{equation}}
\newcommand{\be}{\begin{equation}}
\newcommand{\ee}{\end{equation}}
\newcommand{\bes}{\begin{eqnarray}}
\newcommand{\ees}{\end{eqnarray}}

\def\ot{{\,\otimes \,}}

\def\act{\rhd}

\def\vphi{{\varphi}}

\newcommand{\one}{\mbox{$1 \hspace{-1.0mm}  {\bf l}$}}

      \def\nn{{\nonumber}}

\def\tr{{\mathrm{tr}}}

\def\act{{\, \triangleright\, }}

\newcommand{\su}{\mathfrak{su}}

\newcommand{\so}{\mathfrak{so}}
\newcommand{\SU}{\mathrm{SU}}

\newcommand{\U}{\mathrm{U}}

\newcommand{\SO}{\mathrm{SO}}

\def\extd{\mathrm {d}}
\newcommand{\e}{\epsilon}

\newcommand\acts\triangleright

\newcounter{letter} \newcounter{numeral} \newcounter{Numeral}

\newcommand\Tr{\mathrm{Tr}}
\def\vphi{\varphi}

\def\e{\mbox{e}}
\def\E{\mbox{E}}
\def\extd{\mathrm {d}}
\def\ve{\varepsilon}

\newcommand\maps{\colon}

\newtheorem{theo}{Theorem}
\newtheorem{lemma}[theo]{Lemma}

\begin{document}

%

%

\title{Bubbles and jackets: new scaling bounds in topological group field theories}

\author{Sylvain Carrozza}\email{sylvain.carrozza@aei.mpg.de}\affiliation{Laboratoire de Physique Th\'{e}orique, CNRS UMR 8627,
Universit\'{e} Paris XI, F-91405 Orsay Cedex, France, EU}\affiliation{Max Planck Institute for Gravitational
Physics, Albert Einstein Institute, Am M\"uhlenberg 1,
14476 Golm, Germany, EU}

\author{Daniele Oriti}\email{daniele.oriti@aei.mpg.de}\affiliation{Max Planck Institute for Gravitational
Physics, Albert Einstein Institute, Am M\"uhlenberg 1,
14476 Golm, Germany, EU}


\begin{abstract}
\bigskip 
\bigskip
We use a reformulation of topological group field theories in 3 and 4
dimensions in terms of variables associated to vertices, in 3d, and
edges, in 4d, to obtain new scaling bounds for their Feynman
amplitudes. In both 3 and 4 dimensions, we obtain a bubble bound
proving the suppression of singular topologies 
with respect to the first terms in the perturbative expansion (in the cut-off). We also
prove a new, stronger jacket bound than the one currently available in
the literature. We expect these results to be relevant
for other tensorial field theories of this type, as well as for group
field theory models for 4d quantum gravity.
\end{abstract}


\maketitle

\section{Introduction}
Group field theories and tensor models \cite{iogft, ProcCapeTown, tensorReview, GFTfluid, vincentTensor} are developing more and more as promising candidates for a microscopic description of quantum
spacetime and gravity. They extend to higher dimensions the idea of space time as emerging from a superposition of cellular complexes, generated as the perturbative expansion of a field theory, that is
realized successfully in 2 dimensions  by matrix models \cite{mm}. In this sense they have the same goal of the dynamical triangulations approach to quantum gravity \cite{DT}. Indeed, their Feynman
expansion generates a sum over d-dimensional complexes weighted by amplitudes that characterize the dynamical content of the specific model considered, while the combinatorial structure of the Feynman
diagrams follows directly from the way field arguments are paired in the (interaction term in the) action. 

In the simplest tensor models \cite{tensor}, the basic object is a rank-d tensor whose indices take values in a finite set of dimension $N$, which can be represented graphically as a (d-1)-simplex
with its $d$ (d-2)-faces labelled by the $d$ indices. The action possesses a combinatorially non-local vertex in which $d+1$ tensors are coupled by tracing over common indices respecting the
combinatorial pairing of faces in a d-simplex built by gluing $d+1$ (d-1)-simplices. The amplitudes of such theory, reflecting the simplicity of the data set, are simple functions of the combinatorial
structure of the underlying simplicial complexes. Still, the models constructed on this simple basis are extremely rich, and a wealth of interesting results have been obtained recently about them
\cite{tensorReview,uncoloring}. These include a detailed understanding of the combinatorial structure of the d-complexes arising in perturbative expansion \cite{tensorReview,uncoloring}, and of the
large-N limit of the models \cite{large-N},  the discovery of their symmetries in such limit \cite{razvanVirasoro}, and interesting universality results \cite{razvanUniversality}.  Many of these
results have been obtained thanks to the introduction of {\it colours} in the formalism, that is for multi-tensor models in which $d+1$ complex tensors (each corresponding to a different "colour") enter
the fundamental interaction term (thus each d-simplex has (d+1) faces of different colours) and no colour change is allowed during propagation (thus d-simplices interact only through faces of the same
colour). Equivalently \cite{uncoloring}, colours can be replaced by a labelling of indices of {\it unsymmetric} complex tensors, which are then restricted to have only interactions that satisfy a
$\otimes_d U(N)$ symmetry (so-called {\it invariant} tensor models \cite{uncoloring}). 

Another set of models is obtained when the domain space for the tensorial objects is chosen to be a Lie group manifold, while keeping the same combinatorial structure of action and Feynman diagrams.
This variety of tensorial models goes usually under the name of {\it group field theories} (GFT). The group structure brings in a new set of tools for the analysis or the construction of the models,
including for example harmonic analysis \cite{mikecarlo} and non-commutative Fourier transforms \cite{majidfreidel, ad_nc}. Most important, as far as quantum gravity is concerned, it brings the tensor
formalism in direct contact with discrete classical and quantum geometry, and with other approaches to quantum gravity, such as loop quantum gravity \cite{lqg} and spin foam models \cite{sf}, as well
as simplicial quantum gravity path integrals \cite{regge, 1stRegge}. This is true, in particular, for GFT models in which the fundamental field is assumed to satisfy an additional symmetry property,
that is an invariance under the diagonal shift of the $d$ group arguments of the field by the same group element. The first result of this imposition is to introduce a coupling between different
arguments of the field (the resulting model is not {\it independently distributed} anymore, in the language of matrix models), which is then reflected into a very different dependence of the
amplitudes of the models on the combinatorial structure of the underlying complex. Another result of this imposition (which can be turned into a modification of the kinetic and/or interaction kernels,
rather than of the field itself) is to introduce a gauge connection at the level of the Feynman amplitudes, which can be interpreted (when the Lie group is chosen appropriately) as a discrete gravity
connection and allows then a similar interpretation for the same group arguments of the GFT field. The boundary states (GFT observables) of the model take then the form of holonomy functionals as in
lattice gauge theory and loop quantum gravity, and their harmonic expansion in group representations turns them into spin networks, identified in loop quantum gravity as candidate quantum states of
geometry. The reformulation of the same GFTs, via non-commutative Fourier transform, in Lie algebra variables brings instead the simplicial geometry of the same models to the forefront and expresses
their Feynman amplitudes as simplicial gravity path integrals \cite{ad_nc, GFT-BC, GFT-Holst}, on top of allowing the identification of the GFT analogue of discrete diffeomorphisms \cite{GFTdiffeos}.

Last, we mention that the simple tensor and GFT models with trivial (pure mass) propagators can be also understood as the static ultralocal limit of {\it dynamical} models which include derivative
operators on the group manifold, and thus are not {\it identically distributed}. These type of models, dubbed {\it tensor field theories} and whose matrix analogue is, for example, the
Grosse-Wulkenhaar model \cite{GW}, were first introduced (with different motivations) in the group field theory context in \cite{generalizedGFT}. They have the advantage of allowing a nice definition
of scales in terms of eigenvalues of the kinetic operators, that give rise to a genuine renormalization group flow. Indeed, they have been argued to arise from the radiative corrections of topological
(and ultralocal) GFT models \cite{ValentinJoseph}, and, more recently, some simple models of this type in both 3 and 4 dimensions (with $U(1)$ group) have been shown to be renormalizable to all orders
in perturbation theory \cite{josephvincent,josephsamary}. 

\

Renormalization is indeed one of the aspects of group field theories that is currently at the centre of much research activity. Establishing renormalizability of interesting models for quantum gravity
is crucial for considering them as well-posed field theories, and a better understanding of their renormalization flow will help characterizing their properties at different scales and in different
regimes. It is also important for testing the critical behaviour of tensor and group field theories and, possibly, realizing the {\it geometrogenesis} hypothesis \cite{GFTfluid,vincentTensor} , that
is the hypothesis that a continuous extended spacetime emerges as a result of a phase transition of an underlying model of this type (as it happens in matrix models\footnote{See also \cite{graphity}
for the first such proposal in a different context, and \cite{leejoao} for the first attempt to extract possible phenomenological consequences of this scenario in early cosmology}).

A crucial ingredient for any study of group field theory renormalization is a good understanding of the scaling behaviour and divergence structure of GFT amplitudes, and thus of the GFT sum over
complexes which includes not only all simplicial regular topologies, but also singular configurations that do not satisfy manifold conditions because they contain conical singularities.  Many power
counting results have been proven recently, with a special focus on topological group field theories \cite{GFTrenorm, josephvincent, large-N, vertex, matteovalentin}. These include: very general power
counting theorems based on the lattice gauge theory formulation of the GFT Feynman amplitudes \cite{matteovalentin}, showing how these amplitudes depend on both topological properties (e.g. the
fundamental group) of the corresponding cellular complex and on the details of its combinatorial structure; the scaling bounds that are at the root of the results on the large-N limit of such models
\cite{large-N}, in which we now know that only a restricted class of complexes corresponding to spherical manifolds dominate; analyses targeting restricted, more focused topological issues like the
relative weight of pseudomanifolds over manifolds in the GFT perturbative expansion \cite{vertex}. Answering these questions is crucial for considering tensor and GFT models as physical models of
quantum space, given that the smooth structure and trivial topology of macroscopic space has to be explained, rather than assumed, in a formalism that includes more general discrete configurations in
its microscopic dynamics; it is also a test for the growing maturity of this field, since it is exactly the difficulties in tackling them that stopped further progress in the early days of tensor
models. 

\

The work presented in \cite{vertex} is interesting for another reason as well. Just as the analysis of \cite{matteovalentin}, it rests entirely on the geometric interpretation  of both the GFT field
and action, and of the resulting Feynman amplitudes. Thus its results depend crucially on the added structures and, most important, on the added symmetries of such GFT models as compared to the simple
tensor models, for which we are not aware of similar results. In particular, it uses a reformulation of the same field theory in which the arguments of the field (corresponding to a triangle in 3d)
label vertices rather than edges of the same triangle, a reformulation that was in turn suggested by the analysis of symmetries in GFT \cite{GirelliLivine}, and especially discrete diffeomorphisms
\cite{GFTdiffeos}, which act indeed as local translations of such vertices. In this reformulation, the GFT amplitudes factorize differently than in edge variables and in particular put to the
forefront the bubble structure of the underlying complex, making the analysis of its possible conical singularities straightforward. 

In this paper we build on this earlier result in two ways. First, we perform again the same construction of \cite{vertex} for the 3d (Boulatov) model \cite{boulatov} in a different regularization,
which allows a simpler proof of the bubble bound, showing generic suppression of pseudomanifolds; for the same model, then, we show how the vertex formulation also allows a proof of new jacket bounds
that end up being actually {\it stronger} than the ones on which the large-N results were based \cite{large-N}. Second, we generalize the analysis to the 4d (Ooguri) model \cite{ooguri}. In this
context, the relevant reformulation, that we detail, involves going from variables associated to the triangles on the boundary of the tetrahedron corresponding to the GFT field to variables
corresponding to its edges. Again this is suggested by the topological symmetry studied at the GFT level in \cite{GFTdiffeos}. Using this reformulation, we can then take advantage of the modified
factorization of the GFT amplitudes to prove the 4d generalization of the above results for the 3d model: a bubble bound showing how singular pseudomanifolds are suppressed with respect to the leading order, 
and a stronger jacket bound. 

It is interesting to point out also that, as we will notice again in the following, the structure of the GFT amplitudes in our vertex/edge reformulation  matches nicely the one obtained from the
effective 1-colour dynamics of the underlying coloured GFT. Finally, we notice that the construction we present can be generalized to higher dimensions, so we expect that interesting bounds in any
dimensions could be recovered recursively, in the same way as we obtain here the
jacket bound in 4d from bounds in 3d; also, because our proofs do not rely on invariance of the amplitudes under so-called dipole moves (that are often crucially used in dealing with tensor model
amplitudes) nor on other forms of topological invariance, they can therefore be expected to be more relevant for non-topological models like those for 4d quantum gravity \cite{GFT-BC, GFT-Holst,
GFT-EPRL}. We return to the possible extensions of our results for GFT models of 4d gravity in the conclusions.

\section{3d case: the Boulatov model}

In this section, we start by recalling the construction of the vertex representation of the Boulatov model, introduced in \cite{GFTdiffeos} and further developed in \cite{vertex}. We then use this
representation to derive various scaling bounds, including bounds on
singularities as in \cite{vertex, vertex_short}, and also jacket bounds, which dictate the $1/N$ expansion of the model \cite{large-N}. Interestingly, our method allows to prove a stronger bound than
the usual one. The decay we will obtain is indeed governed by the maximal genus possessed by the jackets of a given graph, as opposed to the average of their genera.
\
The second improvement with respect to \cite{vertex} is due to the use of a heat-kernel regularization instead of sharp cut-offs, which improves the clarity of the presentation and simplifies the
derivation of our results.
 
 \subsection{Vertex representation}


The coloured Boulatov model \cite{cboulatov} can be defined by the action:
\bes \label{action_boulatov}
S[\vphi_\ell]&=& S_{kin} [\vphi] + S_{int}[\vphi] ,\\
S_{kin} [\vphi_\ell] &=& \frac{1}{2} \int [\extd g_i]^3 \sum_{\ell=1}^4   \, \vphi_\ell(g_1, g_2, g_3) \overline{\vphi_{\ell}}(g_1, g_2, g_3) ,\\
S_{int}[\vphi_\ell] &=& \lambda \int [\extd g_{i} ]^6 \, \vphi_1(g_1, g_2, g_3) \vphi_2(g_3, g_5, g_4) \vphi_3(g_5, g_2, g_6) \vphi_4(g_4, g_6, g_1)
\nn \\
&+& \overline{\lambda} \int [\extd g_{i} ]^6 \, \overline{\vphi_1}(g_1, g_2, g_3) \overline{\vphi_2}(g_3, g_5, g_4) \overline{\vphi_3}(g_5, g_2, g_6) \overline{\vphi_4}(g_4, g_6, g_1).
\ees
The four complex fields $\vphi_\ell$, labelled by a colour index $\ell \in \{ 1, \cdots , 4\}$, are defined on $\SU(2)^{3}$. They are interpreted as quantized triangles, and their $\SU(2)$ variables as
parallel transports of an $\su(2)$ connection along paths going from the centre of each triangle to the centre of one of its edges. This geometric interpretation is implemented via the following gauge
invariance under the diagonal action of $\SU(2)$:
\beq\label{gauge_v}
\forall h \in \SU(2),  \qquad \vphi(hg_1, hg_2, hg_3)  \, = \, \vphi(g_1, g_2, g_3) .
\eeq
The action itself has the following geometrical interpretation: the interaction terms encode the gluing of four triangles along their edges so as to form a tetrahedron, and the kinetic terms identify
triangles of the same colours (see Fig. \ref{propa_vertex_v} for pictorial representation as stranded diagrams). In the perturbative expansion of the models, these two elementary operations will
generate simplicial complexes labelling the Feynman amplitudes, which we will wish to
interpret as discrete spacetimes. 

\begin{figure}[h]
  \centering
  \subfloat[Kinetic term]{\label{propa_edge_v}\includegraphics[scale=0.5]{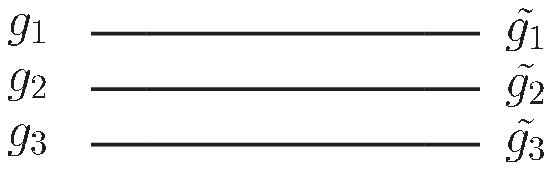}}                
  \subfloat[Interaction term]
{\label{interaction_edge_v}\includegraphics[scale=0.5]{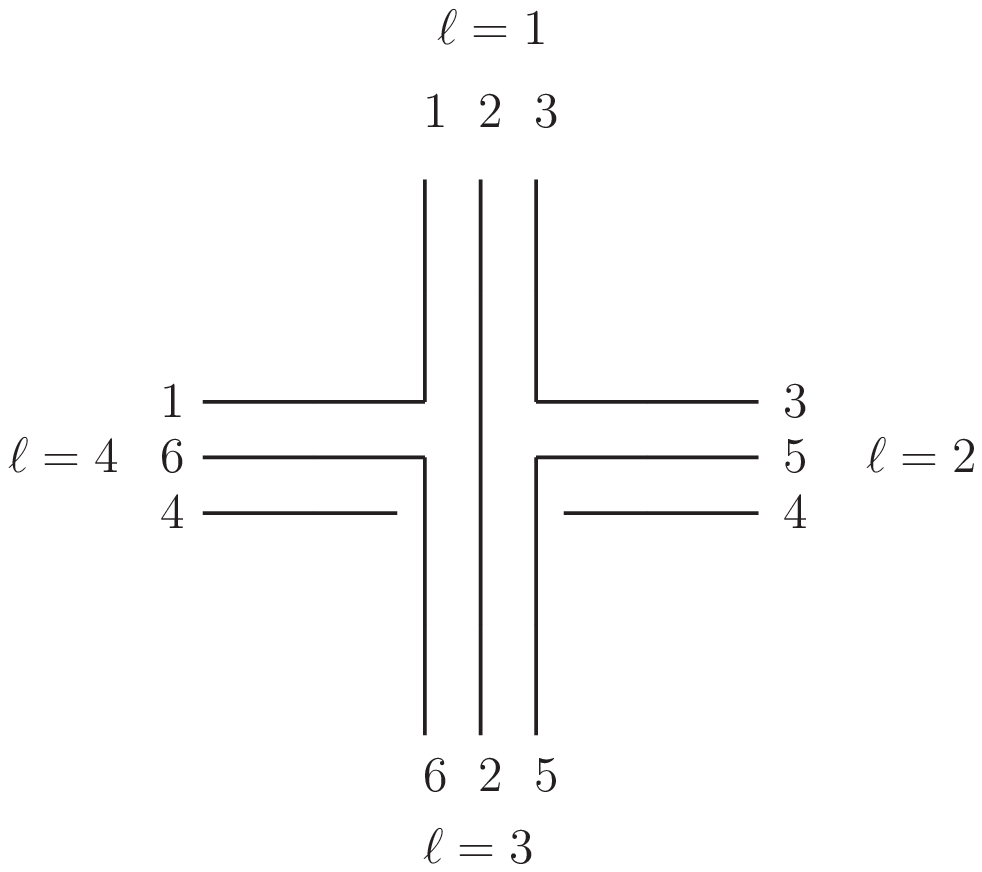}}
  \caption{Kinetic and interaction kernels represented as stranded graphs.}
  \label{propa_vertex_v}
\end{figure}

\

For clarity of the presentation, we will restrict in the following to functions on $\SO(3) \sim \SU(2) / \mathbb{Z}_{2}$, identified as functions $f$ on $\SU(2)$ such that $f(g) = f(-g)$ for all $g \in
\SU(2)$. This will allow us to use a group Fourier transform \cite{majidfreidel}, which is bijective in the case of $\SO(3)$, but not in the case of $\SU(2)$. Note however that we could work with the
full $\SU(2)$
group, at the price of using the generalised Fourier transform introduced in \cite{karim_fourier}. As this would introduce heavy notations without changing the conclusions of the present paper, we
refrain
from doing so and work within the simplified framework.

\
The group Fourier transform for $\SO(3)$ maps functions $f \in L^{2}(\SO(3))$ to functions $\hat{f}$ defined on the Lie algebra $\so(3) \sim \mathbb{R}^{3}$:
\beq
\widehat{f}(x) \equiv \int \extd g f(g) \e_{g}(x)\,, 
\eeq 
where $\e_g$ are plane-waves, i.e maps from $\so(3)$ to $\U(1)$. There is no canonical choice for these maps, since they rely on a choice of coordinates on the group. We will adopt the following,
frequently used, definition: 
\beq
\forall g \in \SU(2) \,, \qquad \e_g \maps \, x \mapsto e^{\rm{i} \, \Tr(x g)}
\eeq 
where $\Tr$ is defined in terms of the usual trace $\tr$ on $2 \times 2$ matrices as: $\Tr(x g) = {\rm{sign}}(\tr g) \tr(xg)$. The image of the Fourier transform is endowed with a non-commutative
structure, given by a $\star$-product, defined on plane-waves as:
\beq
\forall g_1\,, g_2 \in \SU(2)\,, \qquad \e_{g_1} \star \e_{g_2} \equiv \e_{g_{1} g_{2}}\,,
\eeq
and extended by linearity to the whole image of the Fourier transform. This $\star$-product reflects the group structure of $\SU(2)$, and has the important property of being dual, under group Fourier
transform, to the convolution of functions in $L^{2}(\SO(3))$, in the sense that:
\beq
\forall f_1\,, f_2 \in L^{2}(\SO(3))\,, \qquad  \left( \widehat{f}_{1} \star \widehat{f}_{2} \right)(x)\,=\, \int dg \left[ \int dh f_1(g h^{-1})\,f_2(h)\right] e_g(x)
\eeq

which also implies that integral of the point-wise product of two functions on the group manifold is dual to the integral over $\mathbb{R}^3$ of the $\star$-product of their Fourier transforms:
\beq
\int \extd g f_{1}(g^{-1}) f_{2}(g) = \int \extd x \left( \widehat{f}_{1} \star \widehat{f}_{2} \right)(x)\,.
\eeq
This duality allows to express the Boulatov model as a non-commutative field theory, with fields defined on three copies of $\so(3)$, and an action with the same combinatorial structure as in group
space, expect that point-wise products are replaced by $\star$-products (under integration). We refer to \cite{ad_nc, GFTdiffeos} for details. 

\

Exploiting symmetry considerations put forward in \cite{GFTdiffeos}, the Boulatov model has been reformulated in terms of group or algebra data associated to the vertices of the quantum triangles
\cite{vertex}, as
opposed to the edges in the usual formulation. One can map any left-invariant field $\phi \in L^{2}(\SO(3)^{3})$ to a function $\Upsilon[\phi]$ of three $\so(3)$ elements, defined as:
\beq
\Upsilon[\phi](v_{21}, v_{13}, v_{32}) \equiv \int \extd g_1 \extd g_2 \extd g_3 \phi(g_1, g_2, g_3) \e_{g_2^\inv g_1}(v_{21}) \e_{g_1^\inv g_3}(v_{13}) \e_{g_3^\inv g_2}(v_{32})\,.
\eeq
The variable $v_{ij}$ is interpreted as a position variable for the vertex shared by the edges $i$ and $j$, and generates a (quantum group) symmetry which, at the level of the amplitudes, induces the
well known discrete diffeomorphism symmetry of discrete $BF$ theory. 

The geometry behind this reformulation, and the resulting symmetry is the following. The intrinsic geometry of a triangle in $\mathbb{R}^3$ is fully characterized, up to rotations, by specifying three
edge vectors constrained to close (this corresponds to the usual formulation of the model in terms of Lie algebra elements associated to edges of the triangle) or the three positions of its vertices.
Of course, a simultaneous translation of all vertices simply corresponds to a change of the origin of the coordinates in $\mathbb{R}^3$ and should be immaterial, as far as the geometry is concerned.
This is indeed the case. In fact, the function $\Upsilon[\phi]$ is itself invariant under simultaneous translations of its variables. 
Because this translation is given by a quantum group acting on products of representations, in order to define it properly we need to specify an ordering of the arguments of the field. We therefore
interpret $\Upsilon[\phi]$ as a tensor product of representations: 
\beq
\Upsilon[\phi] = \int \extd g_1 \extd g_2 \extd g_3 \phi(g_1, g_2, g_3) \e_{g_2^\inv g_1} \ot \e_{g_1^\inv g_3} \ot \e_{g_3^\inv g_2}\,.
\eeq
Notice that specifying the ordering of the three arguments of the field is itself equivalent to specifying a colouring of the fields entering the GFT action, and thus of the stranded diagrams generated
by it \cite{uncoloring}.
Introducing the translation $T_{\ve}$, defined on tensor products of plane-waves by:
\beq
\cT_{\ve} \act \left( \e_{g_1}(x_1) \ot \cdots \ot \e_{g_N}(x_N ) \right) \equiv \bigstar_{\ve} \left( \e_{g_1}(x_1 + \ve) \ot \cdots \ot \e_{g_N}(x_N + \ve) \right) \equiv \e_{g_1 \cdots g_N}(\ve)\,
\left( \e_{g_1}(x_1) \ot \cdots \ot \e_{g_N}(x_N) \right)\,,
\eeq
we immediately verify that:
\beq
\cT_{\ve} \act \Upsilon[\phi](v_{21}, v_{13}, v_{32}) = \Upsilon[\phi](v_{21}, v_{13}, v_{32})\,.
\eeq
One can prove that the transformation map $\Upsilon$ is a bijection between the space of left-invariant functions in $L^{2}(\SO(3)^{3})$ and its image, that is the set of functions of three $\so(3)$
elements invariant under
$\cT_{\ve}$. We will not detail this point here, but the analogous transformation will be detailed for the 4d case in the next section. This property ensures that
the theory can fully be formulated in vertex variables. The net result can be expressed with group variables $G_{ij} \equiv g_{i}^{\inv}
g_j$, Fourier duals of the $v_{ij}$. In this space, the configuration fields are distributions $\widetilde{\psi}_\ell$ of the form:
\beq
\widetilde{\psi}_\ell (G_1, G_2, G_3) = \delta(G_1 G_2 G_3) \psi_\ell(G_1, G_2, G_3)\,, 
\eeq
where $\psi_\ell$ are regular functions. We refer to \cite{vertex} for more details of this reformulation. In terms of the newly defined fields, the action takes the form:
\bes
S[\psi] &=&  \frac{1}{2} \sum_{\ell} \int [\extd G^\ell_i] \psi_\ell(G^\ell_1, G^\ell_2, G^\ell_3) \delta(G^\ell_1 G^\ell_2 G^\ell_3) \overline{\psi}_\ell(G^\ell_1, G^\ell_2, G^\ell_3) \nn \\
&+& \lambda \int [\prod_{l \neq l'} \extd G^l_{l'}] \cV(G^l_{l'}) \psi_1^{234} \psi_2^{431} \psi_3^{412} \psi_4^{132} \label{action_vertex} \\
&+& \overline{\lambda} \int [\prod_{\ell \neq \ell'} \extd G^l_{l'}] \cV(G^l_{l'}) {\overline{\psi}}_1^{\lower0.1in \hbox{\footnotesize234}} {\overline{\psi}}_2^{\lower0.1in \hbox{\footnotesize431}}
{\overline{\psi}}_3^{\lower0.1in \hbox{\footnotesize412}} {\overline{\psi}}_4^{\lower0.1in \hbox{\footnotesize 132}}\nn ,
\ees
with the notation convention $\psi_{\ell}^{ijk} \equiv \psi_{\ell}(G^{\ell}_i, G^{\ell}_j, G^{\ell}_k)$, and a vertex function defined by:
\bes
\cV(G^l_{l'}) &=& \delta(G^{1}_{2} G^{1}_{3} G^{1}_{4})\delta(G^{2}_{4} G^{2}_{3} G^{2}_{1})\delta(G^{3}_{4} G^{3}_{1} G^{3}_{2})\delta(G^{4}_{1} G^{4}_{3} G^{4}_{2}) \nn \\
		&& \delta(G^{4}_{2} G^{3}_{2} G^{1}_{2})\delta(G^{4}_{3} G^{1}_{3} G^{2}_{3}) \delta(G^{1}_{4} G^{3}_{4} G^{2}_{4}). 
\ees 
The notations in the vertex function are as follows: upper indices label triangles or equivalently fields, and lower indices are associated to vertices, with the convention that a colour $\ell$ labels
the vertex opposite to the triangle of the same colour $\ell$ (see Fig. \ref{color_conventions_v}). 

Notice the appearance of a kernel of the kinetic term (i.e. the distributional part of the fields $\widetilde{\psi}_\ell$). This is the Fourier dual of the translation invariance of the fields we
described in the Lie algebra
representation. It is interpreted as a consistency constraint on the three group elements associated to a quantized triangle: their ordered product needs by construction to be trivial. 

The vertex
function consists (first line) in the same distributional factors of the fields $\widetilde{\psi}_{\ell}$, imposing the mentioned consistency conditions, while the second lines are flatness conditions
associated to paths around the vertices of the tetrahedra, hence guaranteeing flatness of the connection in the boundary of the tetrahedron (see Fig. \ref{tetrahedron_v}). Note that only three of
these flatness constraints appear in the interaction kernel, while a tetrahedron has four vertices. This is
obviously because only three of these flatness constraints are independent. We can therefore choose any triplet of these four constraints\footnote{The fourth one, associated to the vertex of colour 1
being
simply $\delta(G^{2}_{1} G^{3}_{1} G^{4}_{1})$.} to express the same distribution, implementing all four constraints. 

\begin{figure}[h]
  \centering
  \subfloat[Coloring of vertices]{\label{color_conventions_v}\includegraphics[scale=0.5]{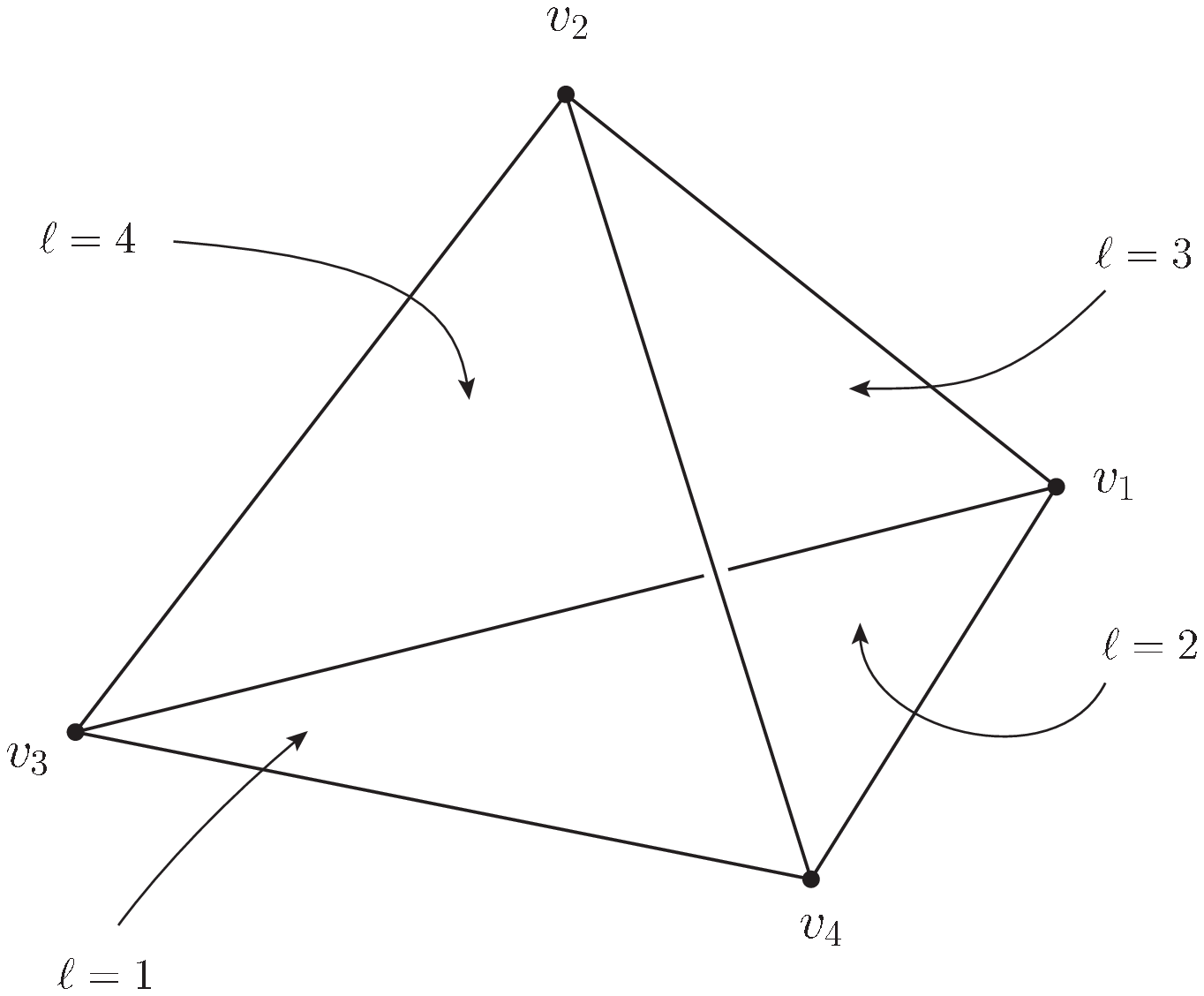}}                
  \subfloat[Holonomy variables]
{\label{variables_v}\includegraphics[scale=0.5]{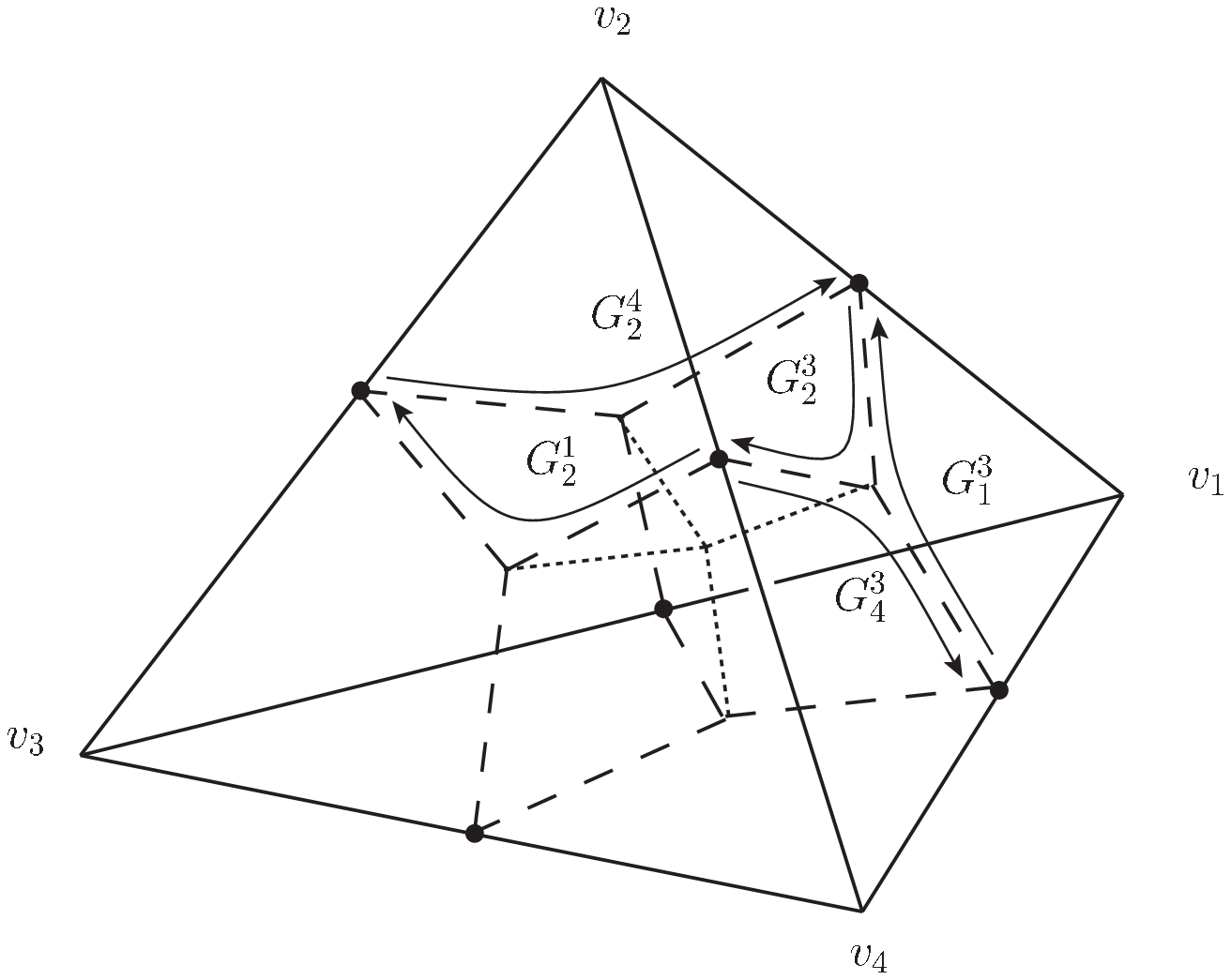}}
  \caption{Coloring conventions and group variables associated to one tetrahedral interaction.}
  \label{tetrahedron_v}
\end{figure}

\

We define the quantum theory via a partition function $\cZ$, defined by a path integral over the space of configuration fields $\widetilde{\psi}_\ell$, weighted by the exponential of (minus) the
action $S$. To make sense of this integral over a space of distributions, we resort to the use of a non Gaussian measure $\mu_{\cP}$ over the space of regular fields $\psi_\ell$ whose covariance is
given by the kernel of the kinetic term: 
\beq
\int \extd \mu_{\cP}(\overline{\psi}, \psi) \, \overline{\psi_\ell(g_1, g_2, g_3)} \psi_{\ell'}(g_1', g_2' , g_3') \equiv \delta(g_1 g_2 g_3) \, \delta_{\ell, \ell'}  \prod_{i = 1}^{6} \delta(g_i
g_i'^{\inv}) \,.
\eeq
This combines the ill-defined measure over the fields $\widetilde{\psi}_\ell$ with the Gaussian part of the integrand into a well-defined measure over the fields $\psi_\ell$. Only the exponential of
the interaction part of the action remains to be integrated, to give a suitable definition of $\cZ$:
\bes\label{partition_vertex}
\cZ &\equiv& \int \extd \mu_{\cP}(\overline{\psi}, \psi) \, \e^{- V[\overline{\psi}, \psi]} \\
V[\overline{\psi}, \psi] &\equiv&  \lambda \int [\extd G] \,\delta(G^{4}_{2} G^{3}_{2} G^{1}_{2})\delta(G^{4}_{3} G^{1}_{3} G^{2}_{3}) \delta(G^{1}_{4} G^{3}_{4} G^{2}_{4})\, \psi_1^{234} \psi_2^{431}
\psi_3^{412} \psi_4^{132}  \; \; + \; \; {\rm{c.c}}.
\ees
A couple of remarks are in order. First, only the flatness part of the kernel of the interaction has been used in the definition of $V$. This is because the distributional nature of the configuration
fields $\widetilde{\psi}_{\ell}$ has already been taken care of in the measure. Were we to integrate $S_{int}$ and not $V$, we would pick up products of equal distributions in the amplitudes, hence
further divergences. Second, at this formal level, which
of the four flatness constraints we use to define $V$ does not matter (see the resulting graphical representation in Fig. \ref{int_vertex_v}). In the regularized theory, to which we turn in the next
subsection, this is not the case, and we expect different choices to give amplitudes with the same scaling behaviour but differing by factors of order $1$ (in the cut-off). That is why, contrary to
the regularization chosen in \cite{vertex}, we use here a symmetric regularization in the colour indices.

\begin{figure}
\centering
\includegraphics[scale=0.5]{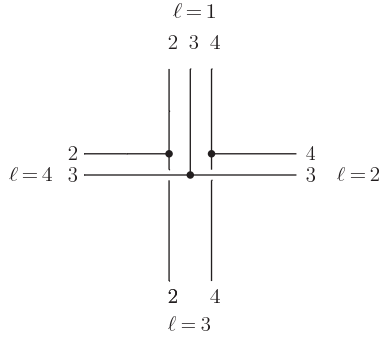}
\caption{Combinatorics of the interaction in vertex variables, in a form suitable for factorization of bubbles of colour $1$.}\label{int_vertex_v}
\end{figure}

\

Formally, we can expand the partition function as a power series in $\lambda \overline{\lambda}$:
\beq
\cZ = \sum_{\cG} \frac{\cA^{\cG}}{\rm{sym}(\cG)}\,.
\eeq
where each amplitude $\cA^{\cG}$ is a monomial in the expansion parameter, labelled by a Feynman graph $\cG$, and $\rm{sym}(\cG)$ is the order of the set of automorphisms of the same graph. 

Before moving on to the analysis of these amplitudes, let us summarize some properties of the diagrams emerging in the above perturbative expansion.
\begin{itemize}
\item The main point of the introduction of colours in GFT is that the Feynman graphs are (in this 3d case) $4$-coloured graphs, which are
well-studied objects in combinatorial topology \cite{FerriGagliardi, Vince_gene, Vince_2d}\footnote{See also \cite{francesco}, in which some of the results are translated in a GFT language.}. We will
call lines and nodes the elementary objects of these graphs, as opposed to edges and vertices which will be used to describe elements of their dual simplicial complexes. 

\item Each line of a $4$-coloured graph is associated to a field of a given colour, and two types of $4$-valent nodes correspond to the two interaction terms in $V$, hence two types of tetrahedra. This
point additionally makes the graphs bipartite, and orientable. 

\item These 4-coloured graphs, encoding gluings of tetrahedra along their boundary triangles are generically dual
to oriented simplicial complexes of a special kind, called pseudomanifolds \cite{r_lost}. They are a generalization of triangulated manifolds, that can be topologically singular at the vertices of
the triangulation, in the sense that the neighbourhood of a vertex does not need to be homeomorphic to a $3$-ball. 

\item Moreover, any lower dimensional simplex of the simplicial complex is dual to a
subgraph of the graph $\cG$. In particular, vertices of colour $\ell$\footnote{Recall that we label both triangles and vertices opposite to them in each tetrahedron with the same colour label $\ell$.}
are in one-to-one correspondence with connected components of the graph obtained from $\cG$ after deleting all the lines of colour
$\ell$. Each such component is called a $3$-bubble (or residue in the mathematical literature). 

\item As a $3$-coloured graph, this $3$-bubble is dual to a triangulated closed surface around the dual vertex.  A property
of special relevance to this paper is that a 3-bubble is a 2-sphere if and only if its dual vertex is regular. Therefore, a simplicial complex generated by our perturbative GFT expansion is a manifold
if and only if all the 3-bubbles of its dual 4-coloured graph are dual to 2-spheres. 

\item In the same way, one generally has a notion of $k$-bubbles, obtained by deleting $(4 -
k)$ colours, and dual to $k$-simplices. Since we are mainly interested in the triangulated surfaces dual to the $3$-bubbles, we will simply call them bubbles throughout the discussion of the 3d model.
\end{itemize}

\

In \cite{vertex}, a factorization of the integrand of the amplitudes in terms of bubble contributions was proven. For any colour $\ell$, one can indeed express the amplitude associated to a graph
$\cG$ of order $\cN$ as:
\beq\label{amplitudes_v}
\cA^{\cG} = (\lambda \overline{\lambda})^{\frac{\cN }{2}}  \int [\extd G]^{\frac{3 \cN}{2}} \left( \prod_{b \in \cB_{\ell}} \prod_{v \in V_{b}} \delta\left( \overrightarrow{\prod_{f \in
\triangle^{b}_{v}}} (G_{v}^{f})^{\epsilon^{f}_{v}}\right)  \right) \left( \prod_{f \in \cF_{\ell}} \delta\left( \overrightarrow{\prod_{v \in f}} G_{v}^{f}\right)\right)\, ,
\eeq
where $\cB_{\ell}$ is the set of bubbles of colour $\ell$, $V_b$ the set of vertices in a bubble $b$, $\triangle^{b}_{v}$ the set of triangles in a bubble $b$ that share one of its vertices $v$, and
finally
$\cF_{\ell}$ is the set of triangles of colour $\ell$ in the complex. 

The geometrical interpretation of this expression is very natural: the bubble terms in the first parenthesis encode flatness around
each of the vertices of the triangulated surface, and the terms in the second parenthesis encode the consistency conditions in triangles of colour $\ell$. 

One thus obtains an effective description
of the model in terms of triangulated 3-cells whose boundaries have colour $\ell$, the bubbles. 

In Fig. \ref{ex_bubble}, we illustrate this result by showing a portion of a bubble $b$, dual to a vertex $v_b$. The effective interaction term associated to $b$ imposes a trivial holonomy around the
vertex $v \in b$.

\begin{figure}[h]
\centering%
\includegraphics[scale=0.5]{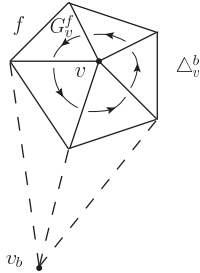}
\caption{Portion of a bubble $b$ dual to $v_b$: in $\triangle^{b}_{v}$, the holonomy around $v$ is imposed to be trivial.}
\label{ex_bubble}
\end{figure}

\

In the regulated theory, this factorization allows to derive nice bounds on the amplitudes: on the one hand, one can
understand the effect of singularities at vertices of colour $\ell$, as was already proven in \cite{vertex}; on the other hand, one can derive so-called jacket bounds, which give information about the
global topology of the 4-coloured graph, and of its dual simplicial complex (in fact, they were crucial to the discovery of the $1/N$ expansion of tensor and GFT models).

  \subsection{Regularization and general scaling bounds}

The regularization we will use in this paper will differ from the one used in the previous work \cite{vertex} in two respects: a) we will only regulate 
what corresponds to the dynamics of $\widetilde{\psi_\ell}$ fields, as we will
show that amplitudes are finite without any regularization of the distributional factors they contain, which are encoded in the measure $\mu_{\cP}$; b) instead of a sharp cut-off in the harmonic
expansion of the $\delta$-functions, we will use heat-kernels. Let us
elaborate a bit on this second point. The $\delta$-function on $\SU(2)$ expands as:
\beq
\delta(g) = \sum_{j \in \frac{\mathbb{N}}{2}} (2 j + 1) \chi_{j}(g) \,.
\eeq
On the other hand, the heat-kernel at time $t > 0$ on $\SU(2)$ is given by:
\beq
\delta_{t}(g) = \sum_{j \in \frac{\mathbb{N}}{2}} (2 j + 1) \e^{- t j(j+1)} \chi_{j}(g) \,,
\eeq 
where $\chi_j$ are the characters of $\SU(2)$. We see that $\delta_t$ converges to the $\delta$-distribution when $t \to 0$, and therefore provides us with a regularization of the $\delta$-function.
The three properties of this regularization we will use in this paper are the following. First, $\delta_t$ is a positive function, which will be very convenient for bounding amplitudes. Second, these
functions behave nicely with respect to the convolution product:
\beq
\int \extd g \delta_{t_1}(g_1 g^{\inv}) \delta_{t_2}(g g_2) = \delta_{t_1 + t_2} (g_1 g_2)\,.
\eeq 
Finally, when $t \to 0$, we have the following asymptotic formula:
\beq
\delta_{t} (g) \underset{t \to 0}{\sim} {\Lambda_{t}}^{3} \,  \e^{- \frac{|g|^2}{4 t}}\,,
\eeq
where $|g|$ is the Riemannian distance between $g$ and the identity, and $\Lambda_{t} \equiv (4 \pi t)^{- 1 / 2}$. In particular we will use the fact that $\delta_{t} (\one) \sim {\Lambda_{t}}^{3}$ at
small $t$.

\

We define the regularized theory by the following partition function:
\bes\label{partition_vertex_reg}
\cZ_{t} &\equiv& \int \extd \mu_{\cP}(\overline{\psi}, \psi) \, \e^{- V_{t}[\overline{\psi}, \psi]} \\
V_{t}[\overline{\psi}, \psi] &\equiv&  \lambda \int [\extd G] \,\frac{\delta_{t}(G^{2}_{1} G^{3}_{1} G^{4}_{1})\delta_{t}(G^{4}_{2} G^{3}_{2} G^{1}_{2})\delta_{t}(G^{4}_{3} G^{1}_{3} G^{2}_{3})
\delta_{t}(G^{1}_{4} G^{3}_{4} G^{2}_{4})}{\delta_{t}(\one)} \, \psi_1^{234} \psi_2^{431} \psi_3^{412} \psi_4^{132}  \;  +  \; {\rm{c.c}}.
\ees

\

Before moving on to the analysis of the regularized amplitudes, let us discuss briefly the role of the regularization chosen.
Note, in fact, that we have chosen a symmetric regularization in the colours, hence the re-introduction of the flatness constraint around the vertex of colour $1$, together with the
appropriate rescaling. The main advantage of such a symmetric regularization is that discussing bounds from different bubble factorizations will be made easier, as will be detailed shortly. 

On the one hand, this regularization is slightly different from the natural scheme one would use in edge variables (see for instance \cite{ValentinJoseph}, which also largely motivated the heat-kernel
regularization used in this paper). This is because it does not correspond to a regularization of the $\delta$-functions imposing flatness of the wedges dual to edges in the GFT interaction vertex
(thus encoding the piece-wise flatness of the simplicial complexes generated in the GFT expansion). On the other hand, regularizations that make bubble factorizations explicit after transformation to
edge variables are very natural also when these last ones are used. For example, we could use a non-symmetric regularization of the interaction 
\beq\label{nonsym_reg}
V_{t}^{1}[\overline{\psi}, \psi] \equiv  \lambda \int [\extd G] \, \delta_{t}(G^{4}_{2} G^{3}_{2} G^{1}_{2})\delta_{t}(G^{4}_{3} G^{1}_{3} G^{2}_{3})
\delta_{t}(G^{1}_{4} G^{3}_{4} G^{2}_{4}) \, \psi_1^{234} \psi_2^{431} \psi_3^{412} \psi_4^{132}  \;  +  \; {\rm{c.c}},
\eeq
which allows to write the amplitude in a factorized form similar to (\ref{amplitudes_v}), with $\ell = 1$. This corresponds exactly, in edge variables, to regularizing only the $\delta$-functions
associated to edges which do not contain the vertex of colour $1$:
\beq
S_{int,1}^{t}[\vphi] = \lambda \int [\extd g ]^9 \,\delta_{t}(g_4 g_4'^{\inv})
\delta_{t}(g_5 g_5'^{\inv}) \delta_{t}(g_6 g_6'^{\inv}) \, \vphi_1(g_1, g_2, g_3) \vphi_2(g_3, g_5, g_4) \vphi_3(g_5', g_2, g_6) \vphi_4(g_4', g_6', g_1)\;  +  \; {\rm{c.c}}.
\eeq
We see, then, that equation (\ref{partition_vertex_reg}) interpolates between the four possible regularizations of this type. 

\

The main fact that makes this form of the regularized amplitude (\ref{partition_vertex_reg}) convenient for our purposes is that, using the positivity of the heat-kernel, one can bound any of the four
flatness constraints by $\delta_{t}(\one)$, and obtain bounds on amplitudes which have the same expression as with a regularization of the non-symmetric type (\ref{nonsym_reg}), but maintaining the
symmetry among colours manifest up to this last step.

One can then show that, for any graph $\cG$ and any choice of colour $\ell_1$, the regularized amplitude $\cA^{\cG}_{t}$ admits the
following bound:
\beq\label{amplitudes_v_reg}
|\cA^{\cG}_{t}| \leq (\lambda \overline{\lambda})^{\frac{\cN }{2}}  \int [\extd G]^{\frac{3 \cN}{2}} \left( \prod_{b \in \cB_{\ell_1}} \prod_{v \in V_{b}} \delta_{\langle v , b \rangle t}\left(
\overrightarrow{\prod_{f \in \triangle^{b}_{v}}} (G_{v}^{f})^{\epsilon^{f}_{v}}\right)  \right) \left( \prod_{f \in \cF_{\ell_1}} \delta\left( \overrightarrow{\prod_{v \in f}}
G_{v}^{f}\right)\right)\, ,
\eeq
where $\langle v , b \rangle$ denotes the number of triangles in a bubble $b$ that contains the vertex $v$.  

\

Let us now repeat how this formula allows to derive interesting scaling bounds (see \cite{vertex} for more details). 

The general idea is that we would like to integrate out or otherwise remove the remaining propagator constraints, which are in a sense non-local quantities (from the point of view of the bubbles) that
prevent us from trivially integrating the amplitudes. 

A simple way of doing it is to pick up a second colour label $\ell_2 \neq \ell_1$, and bound all the flatness constraints associated to vertices of colour $\ell_2$  by their value at the identity. 

In the resulting bound, all
the propagator constraints will then have an independent variable of colour $\ell_2$, allowing us to trivially integrate them. 

We are finally left with two $\phi^{3}$ graphs, corresponding to the strands in
the two remaining colours $\ell_3$ and $\ell_4$. 

Now, each connected component of such a graph is dual to a vertex (of the same colour) of the simplicial complex. Therefore, integrating a tree in each of these components,
and bounding the final expression by its value at the identity, we arrive at:
\beq
|\cA^{\cG}_{t}| \leq (\lambda \overline{\lambda})^{\frac{\cN }{2}} 
\left( \prod_{b \in \cB_{\ell_1}} \prod_{v \in V_{b}(\ell_2)} \delta_{\langle v , b \rangle t}\left( \one \right)  \right)
\left( \prod_{v \in V_{\ell_3} \cup V_{\ell_4}} \delta_{ |v| t}\left( \one \right) \right)\,,
\eeq
where for any $v \in V_{\ell_3} \cup V_{\ell_4}$
\beq
|v| \equiv \sum_{b \in \cB_{\ell_1}\, , \, b \supset v} \langle v , b \rangle\,.
\eeq
is equal to the number of tetrahedra in the simplicial complex that contain $v$. Finally, remarking that:
\beq
\forall a > 0, \qquad \frac{\delta_{at}(\one)}{\delta_{t}(\one)} \underset{t \to 0}{\longrightarrow} a^{- 3/2}
\eeq 
we can rewrite this bounds using powers of heat-kernels with the same parameter, for instance $t$. This allows to show that for any constant $K$ such that
\beq
K > K_0 \equiv \left( \prod_{b \in \cB_{\ell_1}} \prod_{v \in V_{b}(\ell_2)} \langle v , b \rangle^{- 3/2}  \right)
\left( \prod_{v \in V_{\ell_3} \cup V_{\ell_4}}  |v|^{- 3/2} \right)
\eeq
we asymptotically have:
\bes\label{bound_v}
|\cA^{\cG}_{t}| &\leq& K \,(\lambda \overline{\lambda})^{\frac{\cN}{2}} \, [\delta_{t}(\one)]^{\gamma} \nn \\
\gamma &=& \sum_{b \in \cB_{\ell_1}} |V_{b}(\ell_2)| + |\cB_{\ell_3}| + |\cB_{\ell_4}| \,.
\ees

This bound and the following combinatorial fact for connected graphs:
\beq\label{lemma_comb}
\forall \ell \neq \ell'\,, \; |\cB_{\ell'}| + |\cB_{\ell}| - {\sum_{b \in \cB_{\ell}}} |V_{b}(\ell')| \leq 1\,,
\eeq
whose proof can be found in \cite{vertex}, will be the two main ingredients of the following section, where we will derive bubble and jacket bounds.

  \subsection{Bubble and jacket bounds}

In this section, we will give two kinds of bounds on the divergence degree $\gamma_{\cG}$ of a connected vacuum graph $\cG$. In this paper, we adopt the following definition for the divergence degree:
\beq
\gamma_{\cG} = \inf\{ \gamma \in \mathbb{R} \, / \,  \lim (\delta_{t}(\one)^{- \gamma} \cA^{\cG}_{t}) < + \infty \}\,.
\eeq 

\subsubsection{Bubble bounds}
Let us first derive (in this new regularization) the bubble bounds also derived in \cite{vertex}. Starting from (\ref{bound_v}), the derivation is straightforward. 

As a first step, we simply apply (\ref{lemma_comb}) with $(\ell , \ell') = (\ell_1 , \ell_3)$ and $(\ell , \ell') = (\ell_1 , \ell_4)$. This
gives:
\bes
\gamma_{\cG} &\leq& 2 - 2 |\cB_{\ell_{1}}| + {\sum_{b \in \cB_{\ell_1}}} \left(|V_{b}(\ell_2)| + |V_{b}(\ell_3)| + |V_{b}(\ell_4)|\right)\\
&=& 2 - 2 |\cB_{\ell_{1}}| + {\sum_{b \in \cB_{\ell_1}}} |V_{b}|\,.
\ees

We then use the definition of the genus of a bubble $b \in \cB_{\ell}$
\beq
2 - 2 g_b \equiv |V_{b}| - |E_{b}| + |F_{b}|\,,
\eeq
and the combinatorial relation between the number of edges and faces
\beq
2 |E(b)| = 3 |F(b)|
\eeq
to write:
\beq
{\sum_{b \in \cB_{\ell_1}}} |V_{b}| = 
{\sum_{b \in \cB_{\ell_1}}} \left(2 - 2 g_b + \frac{|F_{b}|}{2}\right) = 
2 |\cB_{\ell_1}| + \frac{\cN}{2} - 2 {\sum_{b \in \cB_{\ell_1}}} g_b \,.
\eeq
Therefore:
\beq
\gamma_{\cG} \leq 2 + \frac{\cN}{2} - 2 {\sum_{b \in \cB_{\ell_1}}} g_b \,.
\eeq
In order to obtain a bound that is uniform in the number of interaction vertices (dual tetrahedra), one can simply rescale the GFT coupling constant (this is the same rescaling used in the large-N
expansion) as
\beq\label{rescaling_v}
\lambda \to \frac{\lambda}{\sqrt{ \delta_{t}(\one) }}\,. 
\eeq

We then summarize our result in the following proposition:
\begin{proposition}
Using the rescaling (\ref{rescaling_v}), the divergence degree $\gamma_{\cG}$ of any connected vacuum graph $\cG$ of the coloured Boulatov model verifies, for any colour $\ell$:
\beq
\gamma_{\cG} \leq 2 - 2 {\sum_{b \in \cB_{\ell}}} g_b \,.
\eeq 
\end{proposition}

These bounds have moreover been shown \cite{vertex}  to be {\it optimal}, in the following sense \cite{vertex}: for any positive integers $(g_1 , \cdots ,g_n)$, there exists a vacuum connected graph
$\cG$ such that its
bubbles of colour $\ell$ have genera $(g_1 , \cdots ,g_n)$, and $\gamma_{\cG} = 2 - 2 \sum_{i = 1}^{n} g_i$.

\

Finally, as an immediate corollary of the previous proposition, on can give a bound in terms of the number of pointlike singularities of a given colour:
\begin{corollary}
With the rescaling of the coupling constant (\ref{rescaling_v}), the divergence degree $\gamma_{\cG}$ of any connected vacuum graph $\cG$ verifies, for any colour $\ell$:
\beq
\gamma_{\cG} \leq 2 (1 - N^{s}_{\ell}) \,,
\eeq
where $N^{s}_{\ell}$ is the number of singular vertices of colour $\ell$. 
\end{corollary}

To summarize once more: in the perturbative expansion of the Boulatov GFT model, singular simplicial complexes are generically suppressed, uniformly in the number of tetrahedra, with respect to
to the leading order, which is populated by regular manifolds only. In particular, they have all convergent amplitudes.

\subsubsection{Jacket bounds}

We now focus on the jacket bounds, which are also crucial to the $1/N$ expansion. We show not only that they can also be deduced from our framework, in a rather straightforward way, but also that the
vertex reformulation of the model allows to derive a stronger bound than the known one. 

\

Jackets are two dimensional closed and orientable surfaces
embedded in a simplicial complex \cite{jimmy}. In the simplicial complex dual to a 4-coloured graph $\cG$, we have three different jackets, each labelled by a pair $(\sigma , \sigma^{\inv})$ of cyclic
permutations of the colour set \cite{FerriGagliardi, Vince_gene}. We can identify them as follows. We assign a colour $(\ell \ell')$ to any edge between two vertices of colours $\ell$ and $\ell'$. Then,
to each
tetrahedron, one associates a rectangle, whose edges are so constructed: for any colour $\ell$, the middle point of the edge of the tetrahedron of colour $(\sigma(\ell) \sigma(\ell + 1))$ is
joined to the middle point of the edge of colour $(\sigma(\ell +1) \sigma(\ell + 2))$\footnote{The additions of colours are of course understood modulo $4$.}. An example is given in Fig. \ref{jacket_v}.
Gluings of tetrahedra induce gluings of
these elementary rectangles, providing a quadrangulated surface: this is a jacket. The construction just outlined also makes clear why there are three possible jackets that can be embedded in the
simplicial complex. Also, one can show that these three jackets correspond to three possible reductions of the Boulatov model to a matrix model by reduction with respect to the diagonal gauge
invariance (closure constraint) at the level of the action \cite{jimmy}.

\begin{figure}[h]
\centering
\includegraphics[scale=0.5]{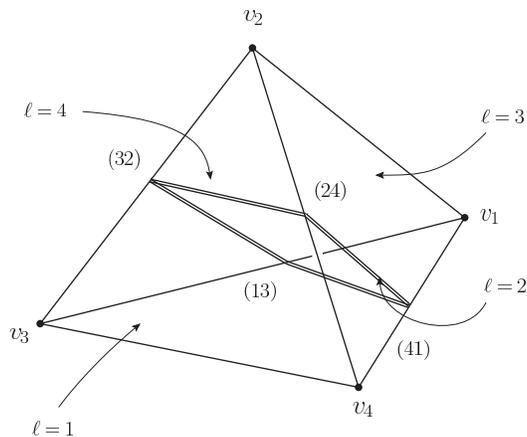}
\caption{Elementary building block of the jacket associated to the cycle $\sigma = (1\,4\,2\,3)$ (double lines).}
\label{jacket_v}
\end{figure}

\
Before moving on to the computation of the bound, we need to express the genus of a jacket $J = (\sigma , \sigma^{\inv})$ in terms of combinatorial quantities of $\cG$. This is the purpose of the
following lemma:
\begin{lemma}
The jacket of colour $J = ( \sigma , \sigma^{\inv} )$ of a coloured graph $\cG$ has genus:    
\beq\label{g_jacket}
g_J = 1 + \frac{1}{2} \left( |T| - \sum_{\ell} |E(\sigma(\ell) \sigma(\ell + 1))| \right) \,,
\eeq
with $T$ and $E(i)$ respectively the sets of tetrahedra and edges of colour $i$ of the simplicial complex dual to $\cG$. 
\end{lemma}
\begin{proof}
The jacket $J$ is an orientable surface, hence its genus is related to its Euler characteristic by:
\beq
2 - 2 g_J = \chi_J = |\cF_J| - |\cE_J| + |\cV_J| \,,
\eeq 
where $\cF_J$,$\cE_J$ and $\cV_J$ are the sets of faces, edges and vertices of $J$. But, by construction:
\beq
|\cF_J| = |T| \;, \qquad |\cE_J| = |t| \;, \qquad |\cV_J| = \sum_{\ell} |E(\sigma(\ell) \sigma(\ell + 1))| \;,
\eeq
where $t$ is the set of triangles in the simplicial complex. Since each triangle is shared by two tetrahedra, we also have $|t| = 2 |T|$, and the result follows.
\end{proof}

\

Let us now consider a connected vacuum graph $\cG$, and one of its jackets $J = (\sigma , \sigma^{\inv})$. We can use the previous lemma to write $g_J$ data suited to bubble factorizations.
Indeed, for any distinct colours $\ell$ and $\ell'$, one immediately has:
\beq
|E(\ell \ell')| = \sum_{b \in \cB_{\ell}} |V_{b}(\ell')|\,, 
\eeq 
since to any bubble $b$ dual to the vertex $v_\ell$, and vertex $v_{\ell'} \in V_{b}(\ell')$, one uniquely associates the edge $(v_{\ell} v_{\ell'}) \in E(\ell \ell')$.
Therefore:
\beq
g_J = 1 + \frac{\cN}{2} - \frac{1}{2} \sum_{\ell} \sum_{b \in \cB_{\sigma(\ell)}} |V_{b}(\sigma(\ell + 1))|  \,.
\eeq

We can now try to make $g_J$ appear in the bounds we computed so far. Applying (\ref{lemma_comb}) to (\ref{bound_v}), with $(\ell , \ell') = (\ell_3 , \ell_3)$, we obtain:
\beq\label{ineq_3d}
\gamma_\cG \leq 1 + {\sum_{b \in \cB_{\ell_1}}} |V_{b}(\ell_2)| + {\sum_{b \in \cB_{\ell_3}}} |V_{b}(\ell_4)|\,.
\eeq
Averaging this expression and
\beq
\gamma_\cG \leq 1  + {\sum_{b \in \cB_{\ell_2}}} |V_{b}(\ell_3)| + {\sum_{b \in \cB_{\ell_4}}} |V_{b}(\ell_1)|\,,
\eeq
we obtain:
\beq
\gamma_\cG \leq 2 + \frac{\cN}{2} - g_J\,, 
\eeq
with $\sigma = (\ell_1 \, \ell_2 \, \ell_3 \, \ell_4)$. As for the bubble bound, a jacket bound that is uniform in the number of GFT interaction vertices (tetrahedra of the simplicial complex) is
obtained by a simple rescaling of the coupling constant, and by the same power of the cut-off used in the case of bubble bound. We
summarize our result in the following proposition: 
\begin{proposition}
With the rescaling of the coupling constant (\ref{rescaling_v}), the convergence degree $\gamma_{\cG}$ of any connected vacuum graph $\cG$ verifies, for {\it any} of its jackets $J$:
\beq
\gamma_{\cG} \leq 2 - g_J \,.
\eeq 

In particular, the following bound holds:

\beq
\gamma_{\cG} \leq 2 - \sup_{J} g_J \,.
\eeq 
\end{proposition}

We note that, as anticipated, this bound is stronger than the usual jacket bound, proven in \cite{large-N}:
\beq
\gamma_{\cG} \leq 2 - \frac{1}{3} \sum_{J} g_J \,.
\eeq 

\

We know already that, if a 3d complex has a jacket with genus zero, the complex is of spherical topology (trivial fundamental group) \cite{tensorReview}. And we know that the above jacket bound is
useful to unravel the topological properties of the various terms appearing in the expansion in powers of the cut-off. 

We conclude by noting also that the same bound could give further insights into the topology of the dominant terms of the same expansion, due to the following fact.

Just as we know that $g = \inf_{\cG , J} g_J$ is a topological invariant, called regular genus, similarly, $\tilde{g} = \inf_{\cG} \sup_{J \in \cG} g_J$ is also well defined, and a topological
invariant
by definition (the $\inf$ is taken over the equivalence class of graphs representing a given topology). If $\tilde{g}$ and $g$ are not identical, then our results allow to derive a non-trivial
topological bound in terms of $\tilde{g}$.

\section{4d case: the Ooguri model}

In this section, we extend the previous results to four dimensions, namely to the (coloured) Ooguri model. Like the Boulatov model, it is a GFT quantization of topological BF theory, and its Feynman
amplitudes can be written either as the simplicial path integral of such theory or as its state sum quantization. Because 4d gravity models are constructed by constraining the data appearing in such
model, either at the level of the GFT action or directly at the level of its Feynman amplitudes, we see the results presented in this section as a first step towards performing a similar analysis in
4d gravity models.

\
We will first show that the coloured Ooguri GFT model, usually formulated in terms of group-theoretic data associated to triangles in $4$-dimensional simplicial complexes \cite{cboulatov}, can be
equivalently written with data associated to edges in the same simplicial complexes. 

Similarly to the Boulatov model, such a formulation will allow to factorize the amplitudes in terms of bubbles (here the $4$-bubbles), and to use new computation tools to derive bounds on the
regularized amplitudes. 

The two main results of this construction will be again: a) a bound on topologically singular vertices, resulting in a clear separation between leading order graphs corresponding to regular manifolds
and sub-dominant graphs associated to non-manifold configurations; b) a new proof and an improvement of the so-called jacket bound \cite{large-N}, which moreover does not rely on topological moves
(dipole contractions). 

\subsection{Edge representation}


\subsubsection{Action and partition function}
The coloured Ooguri model is a field theory of five complex scalar fields $\{\vphi_\ell\ , \ell=1, \ldots ,5\}$, each of them defined over four copies of $\SO(4)$, which respect the following 
gauge invariance condition:
\beq \label{gauge}
\forall h \in \SO(4),  \qquad \vphi_\ell(hg_1, hg_2, hg_3, hg_4)  \, = \, \vphi_\ell(g_1, g_2, g_3, g_4) .
\eeq
Like in three dimensions, they are interpreted as quantized building blocks of spatial geometry, here tetrahedra. The $\SO(4)$ variables are interpreted as parallel transports of an 
$\SO(4)$ connection from the centre of the tetrahedra to the centres of their boundary triangles. The action encodes the gluing of five tetrahedra to form 
a four-simplex via the interaction term, while the kinetic parts mimic the identification of two tetrahedra along their boundary triangles: 
\bes \label{action}
S[\vphi]&=& S_{kin} [\vphi] + S_{int}[\vphi] ,\\
S_{kin} [\vphi] &=& \frac{1}{2} \int [\extd g_i]^4 \sum_{\ell=1}^5   \, \vphi_\ell(g_1, g_2, g_3, g_4) \overline{\vphi_{\ell}}(g_1, g_2, g_3, g_4) ,\\
S_{int}[\vphi] &=& \lambda \int [\extd g_{i} ]^{10} \, \vphi_1(g_1, g_2, g_3, g_4) \vphi_2(g_4, g_5, g_6, g_7) 
\vphi_3(g_7, g_3, g_8, g_9) \vphi_4(g_9, g_6, g_2, g_{10}) \vphi_5(g_{10}, g_8, g_5, g_1) \nn \\
& & + \; \; {\rm{c.c}}. 
\ees
A graphical representation of the two terms of this action is given in Fig. \ref{vertex_propa_e}.


\begin{figure}[h]
  \centering
 \includegraphics[scale=0.5]{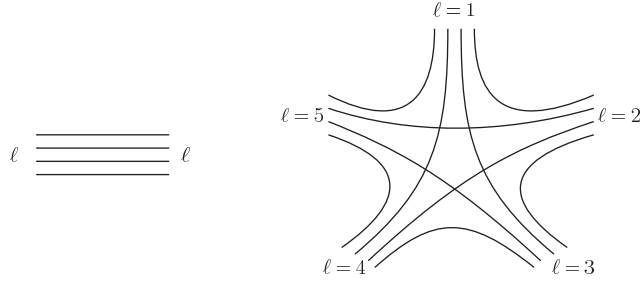}
  \caption{Combinatorics of the kinetic (left) and interaction (right) kernels in the usual (triangle) formulation.}
  \label{vertex_propa_e}
\end{figure}

There is also a metric representation \cite{ad_nc} of the same model in terms of Lie algebra variables, obtained via a group Fourier transform, this time for the group $\SO(4) \simeq ( \SU(2) \times
\SU(2) ) /
\mathbb{Z}_2$. As in the 3d case, we will however restrict ourselves to $\SO(3) \times \SO(3)$, for which a simple invertible group Fourier transform is available, and the generalised framework
\cite{karim_fourier} is not needed. As before, functions on $\SO(3)$ will be identified with functions $f$ on $\SU(2)$ such that $f(g) = f(-g)$. We adopt the notation $g = (g^{+}, g^{-}) \in \SU(2)
\times \SU(2)$,
and similarly for Lie algebra elements, and introduce the plane-waves
\beq
\forall g \in \SU(2) \times \SU(2) \,, \qquad \E_g \maps \, (\su(2) \times \su(2)) \ni x \mapsto \e_{g^{+}}(x^{+}) \, \e_{g^{-}}(x^{-})\,.
\eeq 
The group Fourier transform is given by
\beq
\widehat{f}(x) \equiv \int \extd g f(g) \E_{g}(x)\,, 
\eeq 
and sends the convolution product on $L^{2}(\SO(3) \times \SO(3))$ to a $\star$-product, defined on plane-waves as:
\beq
\forall g_1\,, g_2 \in \SU(2) \times \SU(2) \,, \qquad \E_{g_1} \star \E_{g_2} \equiv \E_{g_{1} g_{2}}\,.
\eeq
Besides these definitions, we will explicitly use the fact that the $\delta$-distributions on $\SO(3) \times \SO(3)$ can be decomposed in plane-waves as:
\beq
\delta(g) = \int \extd x \E_{g}(x)\,.
\eeq 
Dually, a non-commutative $\delta$-distribution on $\so(3) \times \so(3)$ can be defined by:
\beq
\delta_{\star}(x) = \int \extd h \E_{h}(x)\,,
\eeq
which verifies, for any algebra function $f$:
\beq
\left( \delta_{\star} \star f \right) (x) = f(0) \delta_{\star}(0)\,.
\eeq

\

In analogy with what has been done in three dimensions, we would like to use the gauge invariance condition to re-express the action in terms of fields whose arguments are associated to simplices of
one dimension less: in this case, from triangles to edges. 

The main difference with the 3d situation however, is that the numbers of edges and triangles in a tetrahedron do not match: a tetrahedron consists in four
triangles,  but six edges. 

At the level of a field $\vphi_\ell(g_1, g_2, g_3, g_4)$, whose variables are associated to 4 different triangles, this translates into the fact that one can construct six independent
edge variables $G_{ij}$ from pairs of triangle variables $g_i$. For example: 
\bes
G_{41} = g_4^{\inv} g_1 \;, &\qquad G_{42} = g_4^{\inv} g_2 \;, &\qquad G_{43} = g_4^{\inv} g_3 \;, \nn \\
G_{12} = g_1^{\inv} g_2 \;, &\qquad G_{23} = g_2^{\inv} g_3 \;, &\qquad G_{31} = g_3^{\inv} g_1 \;.
\ees
This means that, in order to match the number of degrees of freedom in the two representations, we will have to use more constraints than in the Boulatov model. These constraints will reflect
geometrical 
conditions on the holonomies in a tetrahedron. Remarking that the variable $G_{ij}$ represents the holonomy from the centre of the triangle $i$ to the centre of the triangle $j$, we see that
for any distinct indices $i$, $j$ and $k$, we have:
\beq
G_{ij} G_{jk} G_{kl} = \one \,,
\eeq 
where from now on we use the notation: $G_{ij} \equiv G_{ji}^{\inv}$. There are a priori four such equations to impose (one for any triplet $\{ i, j, k\}$, i.e. one for any vertex of the tetrahedron).
However, only three of them are independent, since for example:
\beq
\left\{
    \begin{array}{lll}
        G_{12} G_{24} G_{41} &=& \one \\
        G_{23} G_{34} G_{42} &=& \one \; \Longrightarrow \; G_{12} G_{23} G_{31} = \one \\
	G_{31} G_{14} G_{43} &=& \one
    \end{array}
\right. 
\eeq

\
This suggests to introduce new fields $\psi_\ell : \SO(4)^{\times 6} \rightarrow \mathbb{C}$, implicitly defined by\footnote{From now on, we will denote by $\delta$ the $\delta$-distribution on
$\SO(3) \times \SO(3)$, and by $\delta^{\SO(3)}$ (resp. $\delta^{\SU(2)}$) its counterpart on $\SO(3)$ (resp. $\SU(2)$).}:
\bes\label{naive}
\vphi_\ell(g_1, g_2, g_3, g_4) &=& \vphi_\ell(G_{41}, G_{42}, G_{43}, \one) \nn \\
 &\equiv& \int \extd G_{12} \extd G_{23} \extd G_{31} 
\delta(G_{12} G_{24} G_{41}) \delta(G_{23} G_{34} G_{42}) \delta(G_{31} G_{14} G_{43}) \\
&\times& \psi_\ell(G_{41}, G_{42}, G_{43}, G_{12}, G_{23}, G_{31})\,. \nn
\ees
This idea can be made precise using the group Fourier transform previously introduced. For any function $\phi \in L^{2}((\SO(3) \times \SO(3))^{\times 4})$, we define a function $\Upsilon[\phi]$ of
six $\so(3)
\times \so(3)$ elements:
\bes
\Upsilon[\phi]( x_{41}, x_{42} , x_{43} , x_{12} , x_{23} , x_{31})
\equiv \int [\extd g_i ] \phi(g_1 , g_2 , g_3 , g_4) && \E_{{g_4}^{\inv} g_1}(x_{41}) \E_{{g_4}^{\inv} g_2}(x_{42}) \E_{{g_4}^{\inv} g_3}(x_{43}) \nn \\ 
&\times& \E_{{g_1}^{\inv} g_2}(x_{12}) \E_{{g_2}^{\inv} g_3}(x_{23}) \E_{{g_3}^{\inv} g_1}(x_{31})\,.
\ees
As in the 3d case, such a function is invariant under a simultaneous (deformed) translation of all its arguments; in this case, however, no geometric interpretation of the variables appearing as field
arguments nor of such invariance can be given, due to the fact that we are not dealing with geometric tetrahedra, but simply with combinatorial simplices to which variables from the classical phase
space of discrete BF theory are associated. The variables do not describe a geometric tetrahedron, and the translation symmetry of each field has nothing to do with the  translation of the vertices
(or the edges) of the tetrahedra in some embedding into $\mathbb{R}^4$; in fact, it is generated by a Lie algebra element of $\so(4)$ and not by a vector in $\mathbb{R}^4$. We note also that the
origin of this symmetry of the GFT field itself, in both this new edge formulation and in the standard triangle formulation, was already pointed out in \cite{GFTdiffeos} and lies in the fact that the
translation/diffeomorphism symmetry of BF theory in 4d (respectively, 3d) (which motivates also our edge (vertex) formulation) is reducible: not all edge (vertex) translations are independent and there is a
trivial global symmetry of the theory under simultaneous translation of all the edges (vertices) of the simplicial complex. In the 4d case, there is an additional reducible component of the edge
translation symmetry, also pointed out in \cite{GFTdiffeos}, which is the one we find at work here: the simultaneous translation of all the (variables associated to the) edges of a tetrahedron that
share a vertex by a Lie algebra element associated to this vertex.

And as in the 3d case, a proper description of the invariance of the field under such deformed translations (in the Lie algebra) requires that we interpret the products of plane-waves as tensor
products, taken in
the order in which we wrote them\footnote{The same remark about the equivalence of specifying an ordering of arguments and colouring, that we made in the 3d model, applies here as well
\cite{uncoloring}.}. 

With this convention in mind, $\Upsilon[\phi]$ is invariant under the following symmetries (only three of them being independent):
\bes
\Upsilon[\phi] &\mapsto& \cT^{142}_{\ve} \act \Upsilon[\phi](x_{ij}) = \bigstar_{\ve} \Upsilon[\phi](x_{41} - \ve, x_{42} + \ve, x_{43} , x_{12} - \ve, x_{23} , x_{31} )\,, \\
\Upsilon[\phi] &\mapsto& \cT^{243}_{\ve} \act \Upsilon[\phi](x_{ij}) = \bigstar_{\ve}
\Upsilon[\phi](x_{41}, x_{42} - \ve, x_{43} + \ve , x_{12} , x_{23} - \ve, x_{31} )\,, \\
\Upsilon[\phi] &\mapsto& \cT^{143}_{\ve} \act \Upsilon[\phi](x_{ij}) = \bigstar_{\ve}
\Upsilon[\phi](x_{41} - \ve, x_{42}, x_{43} + \ve, x_{12}, x_{23} , x_{31} + \ve)\,, \\
\Upsilon[\phi] &\mapsto& \cT^{123}_{\ve} \act \Upsilon
[\phi](x_{ij}) = \bigstar_{\ve}
\Upsilon[\phi](x_{41}, x_{42}, x_{43} , x_{12} + \ve, x_{23} + \ve , x_{31} + \ve)\,.
\ees
These transformations correspond to a simultaneous translation of the Lie algebra variables associated to three edges sharing a vertex in a quantum tetrahedron by the same Lie algebra variables
associated to such common vertex. For instance, $\cT^{142}$ translates the three edges sharing the vertex
of colour $3$.

\
Let us call this space of such invariant fields $\mathbb{T} \equiv {\rm{Im}}(\Upsilon)$, and $\mathbb{D}={\rm{Inv}}((\SO(3) \times \SO(3))^{\times 4})$ the space of fields in $L^{2}((\SO(3) \times
\SO(3))^{\times 4})$ that satisfy the
gauge invariance (\ref{gauge}). We now prove that the map from usual variables to the Lie algebra edge variables is one-to-one.
\begin{proposition}
$\Upsilon$ is a bijection between $\mathbb{D}$ and $\mathbb{T}$. Its inverse maps any $\widetilde{\phi} \in \mathbb{T}$ to: 
\bes
\Upsilon^{\inv}[\widetilde{\phi}](g_i)
\equiv \int [\extd x_{ij}]^{3} 
\left( E_{g_1^{\inv} g_4}(x_{41}) E_{g_2^{\inv} g_4}(x_{42}) E_{g_3^{\inv} g_4}(x_{43}) E_{g_2^{\inv} g_1}(x_{12}) E_{g_3^{\inv} g_2}(x_{23}) E_{g_1^{\inv} g_3}(x_{31})\right) \star
\widetilde{\phi}(x_{ij})\,,
\ees
where only three $x_{ij}$ are being integrated, the others being fixed to any value.
\end{proposition}
\begin{proof}
Let us call $\widetilde{\Upsilon}$ the map defined by the previous formula, and show that $\widetilde{\Upsilon} \circ \Upsilon$ and $\Upsilon \circ \widetilde{\Upsilon}$ are the identity.

\
We first choose $\phi \in \mathbb{D}$, and check that $\widetilde{\Upsilon} \circ \Upsilon [\phi] = \phi$. Using the definitions, we immediately have:
\bes
\widetilde{\Upsilon} \circ \Upsilon [\phi](g_i) = \int [\extd g_{i}'] \phi(g_i') 
\int [\extd x_{ij}]^{3} E_{g_1^{\inv} g_4 g_4'^{\inv} g_1'}(x_{41}) E_{g_2^{\inv} g_4 g_4'^{\inv} g_2'}(x_{42}) E_{g_3^{\inv} g_4 g_4'^{\inv} g_3'}(x_{43}) \\
E_{g_2^{\inv} g_1 g_1'^{\inv} g_2'}(x_{12}) E_{g_3^{\inv} g_2 g_2'^{\inv} g_3'}(x_{23}) E_{g_1^{\inv} g_3 g_3'^{\inv} g_1'}(x_{31})
\ees
The integration over the $x_{ij}$ give three $\delta$-functions. For example, if we choose $x_{41}$ , $x_{42}$ and $x_{43}$ as integrating variables, we obtain:
\bes
\widetilde{\Upsilon} \circ \Upsilon [\phi](g_i) = \int [\extd g_{i}'] \phi(g_i') 
 \delta(g_1^{\inv} g_4 g_4'^{\inv} g_1') \delta(g_2^{\inv} g_4 g_4'^{\inv} g_2') \delta(g_3^{\inv} g_4 g_4'^{\inv} g_3') \\
E_{g_2^{\inv} g_1 g_1'^{\inv} g_2'}(x_{12}) E_{g_3^{\inv} g_2 g_2'^{\inv} g_3'}(x_{23}) E_{g_1^{\inv} g_3 g_3'^{\inv} g_1'}(x_{31})\,.
\ees
We remark that the three $\delta$-functions impose that $g_i g_i'^{\inv}$ is independent of $i$, therefore the three remaining plane-waves are equal to $1$. We can finally introduce a resolution of
the identity $1 = \int \extd h \delta(h g_4 g_4'^{\inv})$, and obtain:
\bes
\widetilde{\Upsilon} \circ \Upsilon [\phi](g_i) &=& \int [\extd g_{i}']  \phi(g_i') \int \extd h \prod_{i = 1}^{4} \delta(h g_i g_i'^{\inv}) \\
&=& \int \extd h \phi(h g_1, h g_2, h g_3, h g_4) \\
&=& \phi(g_1, g_2, g_3, g_4)\,.
\ees
Note that we used the gauge invariance of $\phi$ in the last line.

\
Now, let us take $\widetilde{\phi} \in \mathbb{T}$, and show that $\Upsilon \circ \widetilde{\Upsilon} [\widetilde{\phi}] = \widetilde{\phi}$. We have:
\bes
\Upsilon \circ \widetilde{\Upsilon} [\widetilde{\phi}](x_{ij}) = \int [\extd x_{ij}']^{3} \int [\extd g_i]
\left( E_{g_1^{\inv} g_4}(x_{41}') E_{g_2^{\inv} g_4}(x_{42}') E_{g_3^{\inv} g_4}(x_{43}') E_{g_2^{\inv} g_1}(x_{12}') E_{g_3^{\inv} g_2}(x_{23}') E_{g_1^{\inv} g_3}(x_{31}')\right) \nn \\
\star \, \widetilde{\phi}(x_{ij}') \E_{{g_4}^{\inv} g_1}(x_{41}) \E_{{g_4}^{\inv} g_2}(x_{42}) \E_{{g_4}^{\inv} g_3}(x_{43}) \E_{{g_1}^{\inv} g_2}(x_{12}) \E_{{g_2}^{\inv} g_3}(x_{23})
\E_{{g_3}^{\inv} g_1}(x_{31})\,. 
\ees
Each integral with respect to a variable $g_i$ gives a noncommutative $\delta$-function involving six different Lie algebra elements. For example, the integral over $g_1$ gives a $\delta_{\star}(
x_{41} - x_{41}' + x_{12}' - x_{12} + x_{31} - x_{31}')$. The three others are $\delta_{\star}( x_{42} - x_{42}' + x_{12} - x_{12}' + x_{23}' - x_{23})$, $\delta_{\star}( x_{43} - x_{43}' + x_{23} -
x_{23}' + x_{31}' - x_{31})$, and $\delta_{\star}( x_{41}' - x_{41} + x_{42}' - x_{42} + x_{43}' - x_{43})$. After integration of variables $x_{4j}'$, one obtains:
\beq
\Upsilon \circ \widetilde{\Upsilon} [\widetilde{\phi}](x_{ij}) = ( \cT^{142}_{x_{12} - x_{12}'} \cT^{243}_{x_{23} - x_{23}'} \cT^{243}_{- x_{31} + x_{31}'} ) \act \widetilde{\phi}(x_{ij})\,,
\eeq
which, thanks to the invariance of the field $\widetilde{\phi}$, ends the proof\footnote{Notice that this definition of the transformation from triangle to edge variables relies on the gauge
invariance of the GFT field and on its imposition as a projection operator (thus involving group averaging of the field over the action of the diagonal $\SO(4)$); the above proof makes use of this
invariance and of the compactness of the group itself. As such it would not go through unmodified in the Lorentzian case, where the relevant group, $\SO(3,1)$ would be non-compact. This would in fact
introduce additional divergences that would need to be regulated. However, the generalization of our reformulation of the theory to the Lorentzian setting seems clearly feasible on
geometric/combinatorial grounds, and we expect it to require only additional care in being defined properly.}.
\end{proof}

\

This proposition ensures that an edge formulation is indeed possible. Moreover, the translation invariances of the fields in $\mathbb{T}$ guarantee that, to construct the edge representation in group
space, one just needs to plug (\ref{naive}) in (\ref{action}). In terms of the new fields $\psi_\ell$, the action can be written:
\bes \label{action_edge}
S_{kin} [\psi] &=& \frac{1}{2} \int [\extd G]^6 \sum_{\ell=1}^5   \, \psi_\ell(G^\ell_{41}, G^\ell_{42}, G^\ell_{43}, G^\ell_{12}, G^\ell_{23}, G^\ell_{31}) 
\overline{\psi_{\ell}}(G^\ell_{41}, G^\ell_{42}, G^\ell_{43}, G^\ell_{12}, G^\ell_{23}, G^\ell_{31}) \nn \\
&& \times \delta(G^\ell_{12} G^\ell_{24} G^\ell_{41}) \delta(G^\ell_{23} G^\ell_{34} G^\ell_{42}) \delta(G^\ell_{31} G^\ell_{14} G^\ell_{43}) ,\\
S_{int}[\psi] &=& \lambda \int [\extd G]^{30} \, \psi_1(G_{25}^{1}, G_{24}^{1}, G_{23}^{1}, G_{54}^{1}, G_{43}^{1}, G_{35}^{1}) \psi_2(G_{31}^{2}, G_{35}^{2}, G_{34}^{2}, G_{15}^{2}, G_{54}^{2},
G_{41}^{2}) \nn \\
&& \psi_3(G_{42}^{3}, G_{41}^{3}, G_{45}^{3}, G_{21}^{3}, G_{15}^{3}, G_{52}^{3}) \psi_4(G_{53}^{4}, G_{52}^{4}, G_{51}^{4}, G_{32}^{4}, G_{21}^{4}, G_{13}^{4}) \psi_5(G_{14}^{5}, G_{13}^{5},
G_{12}^{5}, G_{43}^{5}, G_{32}^{5}, G_{24}^{5}) \nn \\
&&\times \left( \prod_{\ell = 1}^{5} \nu(G^{\ell}) \right) \delta(H_{345}) \, \delta(H_{514}) \, \delta(H_{125}) \, \delta(H_{123}) \, \delta(H_{234}) \, \delta(H_{253}) \nn \\
&+& \; \; {\rm{c.c}}, \nn
\ees
where of course $G_{ij}^\ell$ is the (group) variable associated to the edge shared by the triangles $i$ and $j$ in the tetrahedron of colour $\ell$ (thus opposite to the vertex of the same colour).
In this formula, $\nu(G^{\ell})$ is the measure factor associated to the field $\ell$ (i.e. a product of three $\delta$-functions as in the kinetic term), and $H_{ijk}$ is defined by:
$H_{ijk} \equiv G_{jk}^{i} G_{ij}^{k} G_{ki}^{j}$. These $H$'s are holonomies around edges in the boundaries of the 4-simplex corresponding to the GFT interaction vertex, in the very same way as the
3d case was giving flatness conditions around vertices in
boundaries of tetrahedra. As remarked earlier, there are ten edges in a 4-simplex, hence ten flatness conditions to impose (one for each choice of triplet of distinct colours $i$, $j$ and $k$).
However,
only six of them are independent, which is why the same number of $\delta$-functions of this type appear in the interaction. For the same reason, one is free to choose any set of six independent
$H_{ijk}$ to write the distribution encoding flatness. As in 3d, this freedom will prove very useful.

\

The partition function is defined through a path integral, with the propagator encoded in the covariance of a Gaussian measure $\mu_{\cP}$:
\bes
\int \extd \mu_{\cP}(\overline{\psi}, \psi) \, \overline{\psi_\ell(g_1, \cdots, g_6)} \psi_{\ell'}(g_1', \cdots, g_6') &\equiv& \nu(g_1, \cdots , g_6) \, \delta_{\ell, \ell'}  \prod_{i = 1}^{6}
\delta(g_i g_i'^{\inv}) \nn \\
&=& \delta(g_4 g_2^{\inv} g_1) \delta(g_5 g_3^{\inv} g_2) \delta(g_6 g_1^{\inv} g_3) \, \delta_{\ell, \ell'}  \prod_{i = 1}^{6} \delta(g_i g_i'^{\inv})\,;
\ees
with respect to which we integrate the exponential of the interaction part of the action\footnote{Note however that the interaction part does not include the constraints $\nu$, since they are already
imposed in the propagator.}:
\bes\label{partition_edge}
\cZ &\equiv& \int \extd \mu_{\cP}(\overline{\psi}, \psi) \, \e^{- V[\overline{\psi}, \psi]} \\
V[\overline{\psi}, \psi] &\equiv&  \lambda \int [\extd G] \, \cV(\{ H \}) \, \psi_1 \psi_2 \psi_3 \psi_4 \psi_5  \; \; + \; \; {\rm{c.c}}.
\ees
We kept variables of integration implicit in $V$, and called $\cV$ the distribution encoding flatness in (\ref{action_edge}). The kernel $\cV$ is represented as a stranded graph in Fig.
\ref{int_e_fig}\footnote{In order to limit the number of crossings, we have reorganized the strands of the different fields.}.

\begin{figure}
\centering
\includegraphics[scale=0.5]{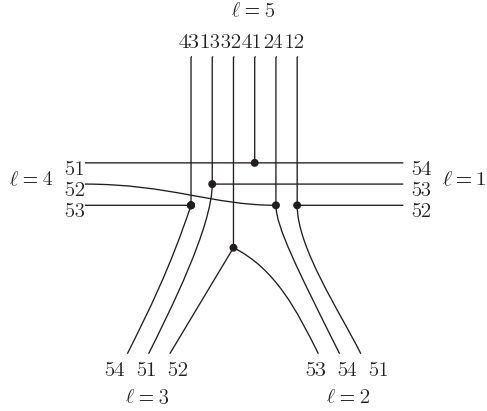}
\caption{Combinatorics of the interaction in edge variables, in a form suitable for factorization of bubbles of colour $5$.}\label{int_e_fig}
\end{figure}

\subsubsection{Amplitudes}
In this section, we explore a similar route as in the 3d case, and propose a way to write amplitudes in such a way that the integrand factorizes into bubble contributions. At this stage, everything is
formal (the amplitudes are generically divergent in absence of some regularization), but will be a precious guide in order to derive bounds once a cut-off is added.

\
We first pick up a given colour (say $\ell=5$) and choose the six edges of the tetrahedron labelled by $5$ to impose the flatness conditions in any given 4-simplex. This corresponds to working with
holonomies $H_{ijk}$
such that one of the indices is equal to 5. We can for example express the distribution $\cV$ in (\ref{action_edge}) as:
\beq
\cV = \delta(H_{345}) \, \delta(H_{514}) \, \delta(H_{125}) \, \delta(H_{315}) \, \delta(H_{245}) \, \delta(H_{253})\,.
\eeq
This expression involves 18 independent variables $G^{\ell}_{i j}$, 6 of them with $\ell = 5$, and 3 for each $\ell \neq 5$. Therefore, when computing the amplitude of a coloured graph $\cG$, all the
propagators with $\ell \neq 5$ will have three strands with free endpoints, that can be integrated straightforwardly. As a result, the closure conditions associated to tetrahedra of colour $5$ will
have trivial contributions. Let us verify it for a tetrahedron of colour $1$. The integral that we have to compute in this case is of the form:
\beq
\cA^{\cG} = \int [\extd G] \delta(G^{1}_{54} G^{1}_{42} G^{1}_{25}) \delta(G^{1}_{43} G^{1}_{32} G^{1}_{24}) \delta(G^{1}_{35} G^{1}_{52} G^{1}_{23}) R(\{ G \}) \,,
\eeq  
where $R$ does not depend on $G^{1}_{42}$, $G^{1}_{23}$ and $G^{1}_{34}$, since $H_{423}$ does not explicitly appears in the expression we chose for $\cV$. Successively integrating these variables we
have: 
\bes
\cA^{\cG} &=& \int [\extd G]  \delta(G^{1}_{43} G^{1}_{32} G^{1}_{25} G^{1}_{54}) \delta(G^{1}_{35} G^{1}_{52} G^{1}_{23}) R(\{ G \}) \nn \\
&=& \int [\extd G]  \delta(G^{1}_{43} G^{1}_{35} G^{1}_{52} G^{1}_{25} G^{1}_{54})  R(\{ G \})\nn \\
&=& \int [\extd G]  R(\{ G \})\,,
\ees
which shows that the amplitude is unchanged if all the propagators of colour $1$ are replaced by trivial ones.
\
Therefore, as in the 3d case, we are lead to simplified expressions for the amplitudes, where all propagators of colour $\ell \neq 5$ are trivially integrated. This allows to factorize the integrand of
the amplitude into bubbles (of colour $5$) contributions. Since within these bubbles all the propagators are effectively trivial, we can contract each connected component of strands in a bubble to one
node. This is easily understood by looking at Fig. \ref{int_e_fig}: in each bubble, all the internal strands are part of lines of colours $\ell \neq 5$, that is those that effectively contain only
three strands; since the constraints associated to these propagators have been integrated with respect to the deleted strands, the remaining strands encode simple convolutions of $\delta$-functions
that can be successively performed. The last $\delta$-distribution in a connected component of strands encodes flatness of its dual edge. We call $\cB_5$ the set of bubbles of colour $5$, $E_b$ the set
of edges in a bubble $b$, $T_5$
the set of tetrahedra of colour $5$. With these notations, the amplitude of $\cG$ (with $\cN$ nodes) can be written as:
\bes\label{amplitude_edge}
\cA^{\cG} &=& (\lambda \overline{\lambda})^{\frac{\cN}{2}} \int [\extd G]^{3 \cN} \left(\prod_{b \in \cB_5}  \prod_{e \in E_b}  \delta( \overrightarrow{\prod_{\tau \supset e}}
(G^{\tau}_{e})^{\epsilon^{\tau}_{e}})\right) 
\left( \prod_{\tau \in T_5} \delta(G^{\tau}_{43} G^{\tau}_{31} G^{\tau}_{14}) \delta(G^{\tau}_{32} G^{\tau}_{21} G^{\tau}_{13}) \delta(G^{\tau}_{24} G^{\tau}_{41} G^{\tau}_{12}) \right) \,,
\ees
where $\epsilon^{\tau}_{e} = \pm 1$ depending on orientation conventions. 

\

The different types of bounds we will prove in the next section will rely on two different ways of trivializing the interaction kernels, so that the integration of the last propagators coupling the
variables appearing in different bubbles in the above expression can be performed. The two strategies give an expression for the amplitudes in which different combinatorial substructures in the
5-coloured graph (and in its dual simplicial complex) are singled out, and will be used to obtain a bound referring to the 4-bubbles in one choice, and a bound referring to jackets in the other.  

In a first strategy, we will trivialize three constraints associated to a tree of edges in each tetrahedron, which will allow to integrate all the remaining propagators. This will lead
straightforwardly to a bubble bound. In a second strategy, we will trivialize the constraints associated to the three edges of a same triangle in each tetrahedron. Only two $\delta$-functions per
propagator will be easily integrable in this case, and the remaining ones will allow to factorize the integrand in terms of Boulatov integrands, henceforth giving a bound involving Boulatov amplitudes
of
$4$-bubbles. This will instead lead to a jacket bound.


  \subsection{Regularization and general scaling bounds}
The amplitudes as written are of course divergent and need to be regularized to be given rigorous meaning.  

A nice aspect of the latter formulation of the Ooguri model lies in the fact that the constraints associated to edges need not to be regularized in order to make the amplitudes finite, as it will be
shortly proven. In other words, only the dynamics of the theory is affected by the cut-off procedure, and not the kinematical space of fields in terms of which the theory is defined. As in the 3d
case, we will use a heat-kernel regularization of the $\delta$-distributions, that with respect to a sharp cut-off will have the main advantage of being positive. The $\SO(3) \times \SO(3)$
$\delta$-distribution splits into a product of two $\SO(3)$ terms: for any $g = (g^{+}, g^{-}) \in \SO(3) \times \SO(3)$, $\delta(g) = \delta^{\SO(3)}(g^{+}) \delta^{\SO(3)}(g^{-})$. Given our
parametrization of the space of functions on $\SO(3) \times \SO(3)$, we can define a regularized distribution for $\SO(3) \times \SO(3)$ using $\SU(2)$ heat-kernels:
\beq
\forall g=(g^{+}, g^{-}) \in \SO(3) \times \SO(3), \qquad \delta_{t}(g) \equiv \delta_{t}^{\SU(2)}(g^{+}) \delta_{t}^{\SU(2)}(g^{-})\,.
\eeq

We therefore define the
regulated partition function as:
\bes\label{partition_edge_reg}
\cZ_{t} &\equiv& \int \extd \mu_{\cP}(\overline{\psi}, \psi) \, \e^{- V_{t}[\overline{\psi}, \psi]} \,.
\ees
$V_{t}$ is the regulated interaction, associated to the kernel:
\bes
\cV_{t} \equiv \frac{1}{[\delta_{t}(\one)]^{4}} \prod_{\{i j k \}} \delta_{t}(G_{jk}^{i} G_{ij}^{k} G_{ki}^{j}) 
= \frac{1}{[\delta_{t}(\one)]^{4}} \prod_{\{i j k \}} \delta_{t}(H_{ijk}) \,   
\ees
where the product runs over the 10 possible choices of 3 colours among 5. Note that we chose a symmetric regularization in the colours, hence the rescaling by $\frac{1}{[\delta_{t}(\one)]^{4}}$. 

The same kind of comments as in $3d$ apply here. First, this
choice of symmetric regularization is convenient as it will allow to easily average over colour attributions. Second, we could have as well chosen a non-symmetric regularization, compatible with the
map $\Upsilon$. For example, a regularized interaction
\bes
S^{t}_{int,5}[\vphi] &=& \lambda \int [\extd g_{i} ]^{10} \, \delta_{t}(g_{2} g_{2}'^{\inv})  \delta_{t}(g_{3} g_{3}'^{\inv})  \delta_{t}(g_{4} g_{4}'^{\inv})  \delta_{t}(g_{6} g_{6}'^{\inv}) 
\delta_{t}(g_{7} g_{7}'^{\inv})  \delta_{t}(g_{9} g_{9}'^{\inv})\nn \\ 
&\times& \vphi_1(g_1, g_2', g_3', g_4') \vphi_2(g_4, g_5, g_6', g_7') 
\vphi_3(g_7, g_3, g_8, g_9') \vphi_4(g_9, g_6, g_2, g_{10}) \vphi_5(g_{10}, g_8, g_5, g_1) 
+ \; \; {\rm{c.c}} \nn \\
\ees
in triangle variables corresponds to a regularized interaction kernel 
\beq
\cV^{5}_{t} \equiv \prod_{\{i j 5 \}} \delta_{t}(H_{ijk})
\eeq
in edge variables, which in turn implies an explicit factorization of the amplitudes in terms of bubbles of colour $5$. This is the exact analogue of (\ref{amplitude_edge}) in the theory with cut-off.

\
With the symmetric regularization, this type of factorization is recovered as a bound only, but for arbitrary colour $\ell$. Moreover this bound will be always saturated in power counting. It
is obtained by bounding four redundant flatness conditions in all the interactions. For instance, if we use the colour $5$ as before, we have:
\beq
| \cV_{t} | \leq \prod_{\{i j 5 \}} \delta_{t}(H_{ij5}) \,,
\eeq 
where now the product runs over the 6 flatness conditions involving the colour $5$. Thanks to the positivity of the regularization, and the convolution properties of the heat-kernel, this immediately
yields:
\bes\label{amplitude_edge_reg}
|\cA^{\cG}_{t}| \leq (\lambda \overline{\lambda})^{\frac{\cN}{2}} \int [\extd G]^{3 \cN} \left(\prod_{b \in \cB_5}  \prod_{e \in E_b}  \delta_{\langle e,b \rangle t}( \overrightarrow{\prod_{\tau
\supset e}} (G^{\tau}_{e})^{\epsilon^{\tau}_{e}})\right) 
\left( \prod_{\tau \in T_5} \delta(G^{\tau}_{43} G^{\tau}_{31} G^{\tau}_{14}) \delta(G^{\tau}_{32} G^{\tau}_{21} G^{\tau}_{13}) \delta(G^{\tau}_{24} G^{\tau}_{41} G^{\tau}_{12}) \right) \quad
\ees
where for any edge $e$ in a bubble $b$, we denote by $\langle e,b \rangle$ the number of tetrahedra in $b$ that contain $e$. 

\

\begin{figure}[h]\label{trees}
  \centering
  \subfloat[A tree of edges.]{\label{tree_tetra1}\includegraphics[scale=0.4]{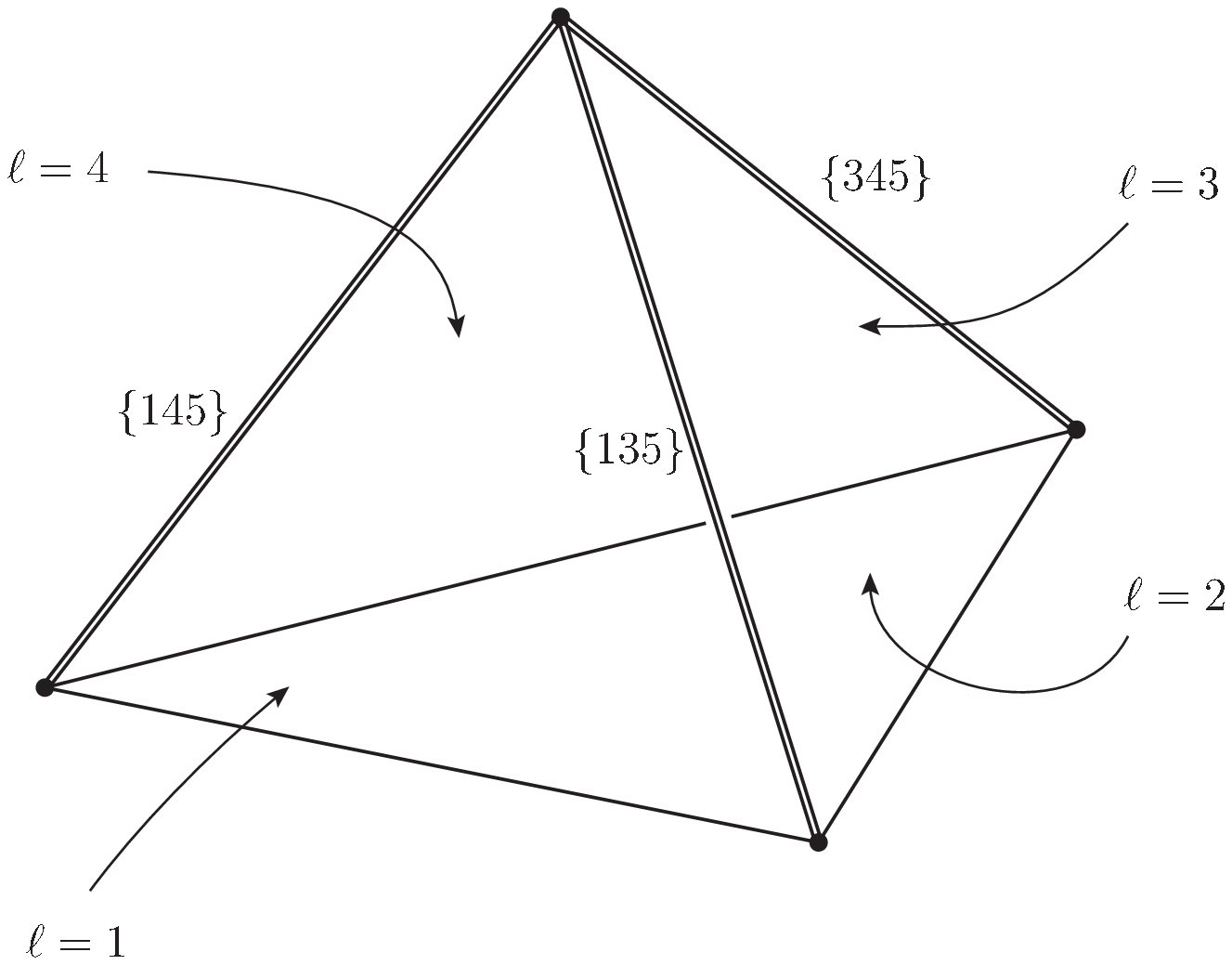}}                
  \subfloat[Another tree.]
{\label{tree_tetra2}\includegraphics[scale=0.4]{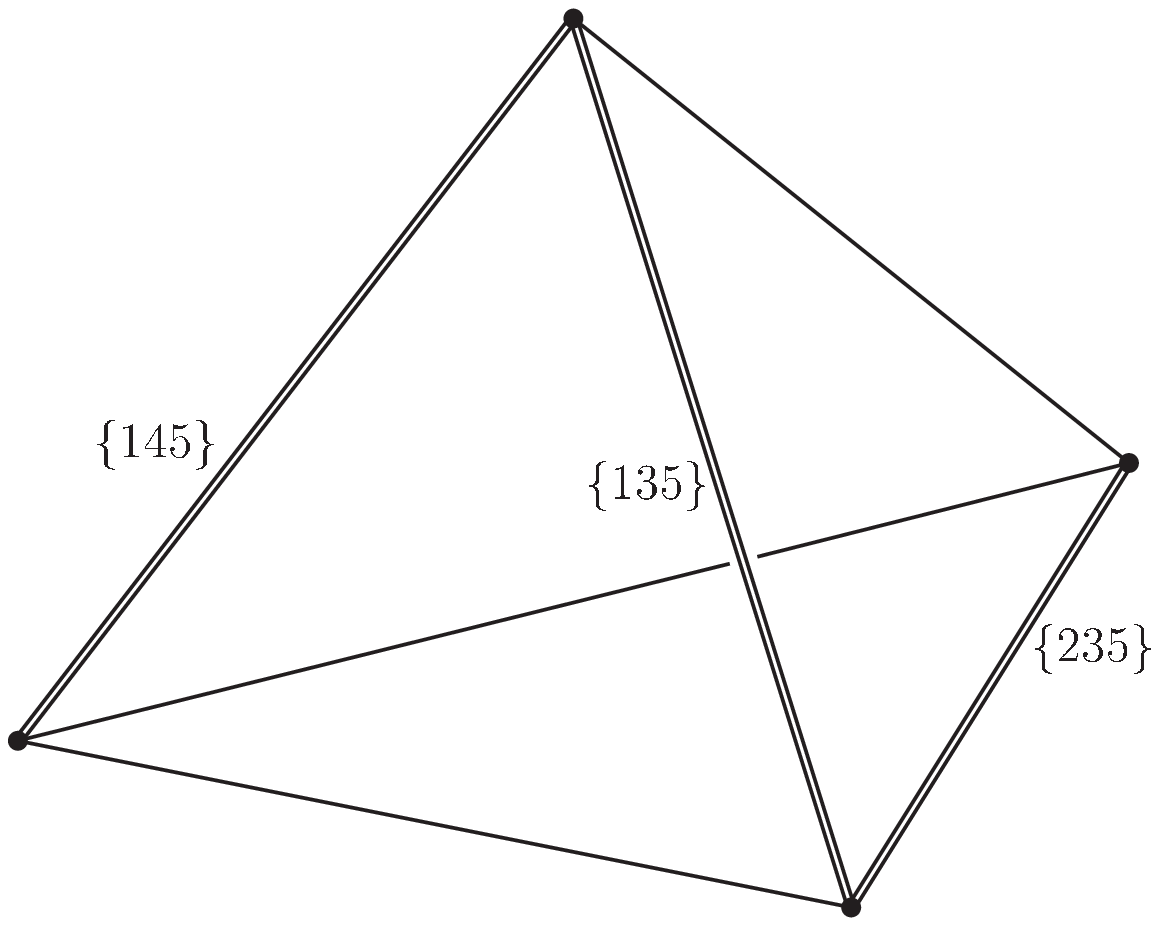}}
\subfloat[Not a tree.]
{\label{not_tree_tetra}\includegraphics[scale=0.4]{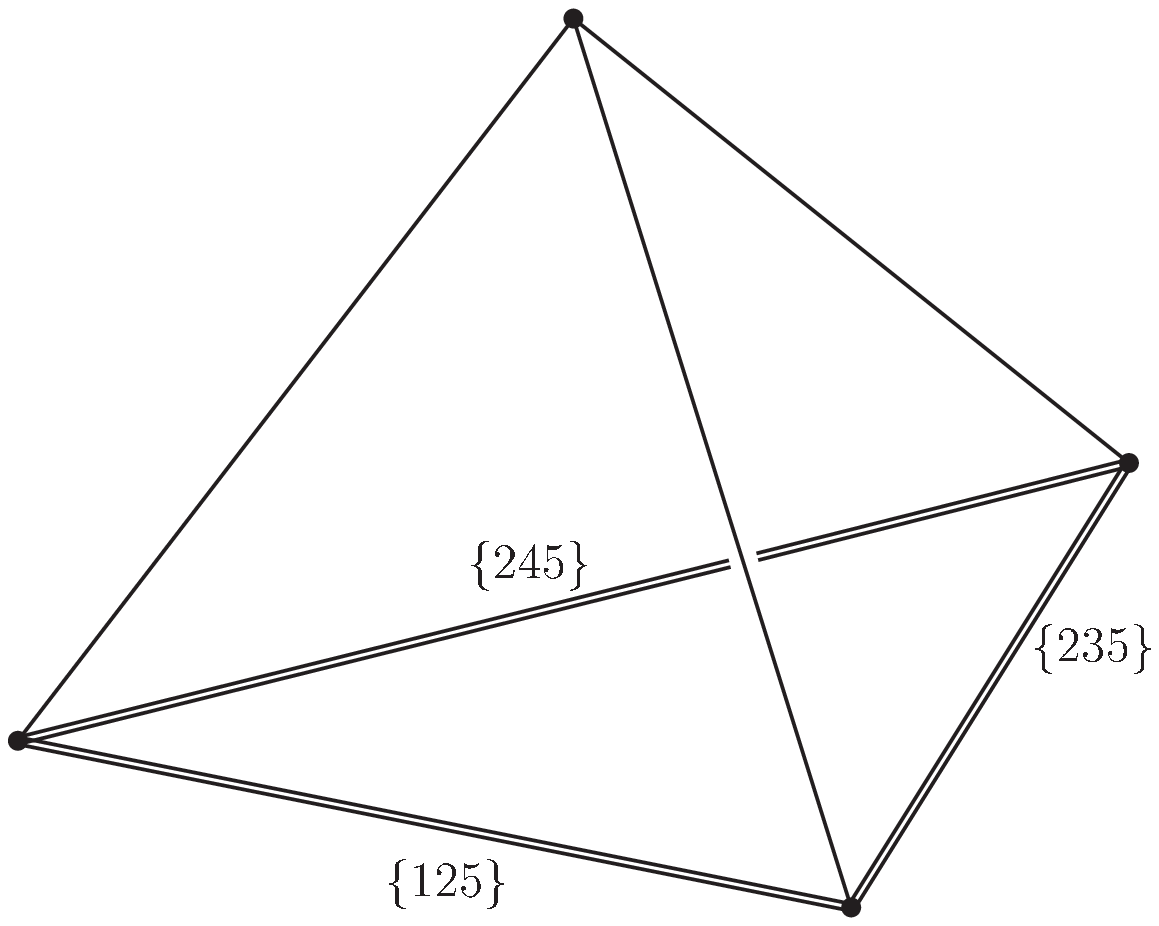}}
  \caption{Different possibilities for trivializing edge interactions in a tetrahedron of colour $5$ (double lines): a tree in Fig. \ref{tree_tetra1} and Fig. \ref{tree_tetra2}; edges associated to a
same triangle in Fig. \ref{not_tree_tetra}.}
\end{figure}

We now have to integrate the remaining propagators, using the two possible sets of variables evoked before. 

\

We start with the first strategy and we look for integrating variables associated to a tree of edges in each tetrahedron. 
There are two kinds of such trees: three edges sharing a same vertex; or three edges such that the first one shares a vertex with the second, the second with the third, but the third does not share
any vertex with the first one (see Fig. \ref{tree_tetra1} and \ref{tree_tetra2}). Of course we want to use the colours to define these trees, so that the same simplification takes place in all the
tetrahedra in the simplicial complex. The notations are as follows: we will associate the colour $\{\ell_1 \ell_2 \ell_3\}$ to an edge involving flatness constraints $H_{\ell_1 \ell_2 \ell_3}$. Such an
edge is therefore dual to a (maximal) connected subgraph of $\cG$, involving only lines of colours $\ell_1$, $\ell_2$ and $\ell_3$.\footnote{Note that this convention differs from the one used in 3d to
label vertices, as here the latter would amount to labelling an edge by the two line colours on which its dual graph does {\it{not}} have support.}

For definiteness, we will use variables involved in flatness conditions around edges of colours $\{345\}$, $\{145\}$ and $\{135\}$, a choice that corresponds to a tree of the first kind (i.e. Fig
\ref{tree_tetra1}). 

Since each strand is connected to two interactions, the integrals to
compute are not simple convolutions, therefore difficult. To circumvent this problem we simply bound all the heat-kernels implementing flatness constraints of colours $\{345\}$, $\{145\}$ and $\{135\}$
by their value at the identity, and then integrate the propagators. This yields:
\bes
|\cA^{\cG}_{t}| &\leq& (\lambda \overline{\lambda})^{\frac{\cN}{2}} \int [\extd G]^{3 \cN} \left(\prod_{b \in \cB_5}  \prod_{e \in E_{b}(345)\cup E_{b}(145) \cup E_{b}(135)}  \delta_{\langle e,b
\rangle t}(\one) \right) \nn \\
&&\times \left(\prod_{b \in \cB_5}  \prod_{e \in E_{b}(125)\cup E_{b}(235) \cup E_{b}(245)}  \delta_{\langle e,b \rangle t}( \overrightarrow{\prod_{\tau \supset e}}
(G^{\tau}_{e})^{\epsilon^{\tau}_{e}})\right) 
\left( \prod_{\tau \in T_5} \delta(G^{\tau}_{43} G^{\tau}_{31} G^{\tau}_{14}) \delta(G^{\tau}_{32} G^{\tau}_{21} G^{\tau}_{13}) \delta(G^{\tau}_{24} G^{\tau}_{41} G^{\tau}_{12}) \right) \nn \\
&\leq& (\lambda \overline{\lambda})^{\frac{\cN}{2}}  \left(\prod_{b \in \cB_5}  \prod_{e \in E_{b}(345)\cup E_{b}(145) \cup E_{b}(135)}  \delta_{\langle e,b \rangle t}(\one) \right)  \nn \\
&& \times \int [\extd G]^{\frac{3}{2} \cN} \left( \prod_{b \in \cB_5}  \prod_{e \in E_{b}(125)\cup E_{b}(235) \cup E_{b}(245)}  \delta_{\langle e,b \rangle t}( \overrightarrow{\prod_{\tau \supset e}}
(G^{\tau}_{e})^{\epsilon^{\tau}_{e}}) \right)\,.
\ees
 
The only term left to integrate is a product of integrals associated to connected $\phi^3$ graphs (whose lines are strands of the initial graph). Each of these is dual to an edge of colour $\{125\}$,
$\{235\}$ or $\{245\}$. Integrating a maximal tree
of strands in each of these graphs, then bounding the last $\delta$-function by its value at the identity, we obtain the general bound:
\beq
|\cA^{\cG}_{t}| \leq (\lambda \overline{\lambda})^{\frac{\cN}{2}}  \left(\prod_{b \in \cB_5}  \prod_{e \in E_{b}(345)\cup E_{b}(145) \cup E_{b}(135)}  \delta_{\langle e,b \rangle t}(\one) \right)
\left( \prod_{e \in E(125)\cup E(235) \cup E(245)}  \delta_{|e| t}( \one) \right)\,,
\eeq
where 
\beq
|e| \equiv \sum_{b \in \cB_\ell \, , \, b \supset e}  \langle e,b \rangle.
\eeq 
coincides with the number of $4$-simplices containing the edge $e$. 

Remarking that, when $t$ goes to zero:
\beq
\forall a > 0, \qquad \frac{\delta_{at}(\one)}{\delta_{t}(\one)} \rightarrow a^{- 3}
\eeq
we can rewrite this bounds using powers of heat-kernels with the same parameter, for instance $t$. This allows to show that for any constant $K$ such that
\beq
K > K_0 \equiv \left(\prod_{b \in \cB_5}  \prod_{e \in E_{b}(345)\cup E_{b}(145) \cup E_{b}(135)}  \langle e,b \rangle^{- 3} \right)
\left( \prod_{e \in E(125)\cup E(235) \cup E(245)}  |e|^{- 3} \right)
\eeq
we asymptotically have:
\bes\label{bound1}
|\cA^{\cG}_{t}| &\leq& K \,(\lambda \overline{\lambda})^{\frac{\cN}{2}} \, [\delta_{t}(\one)]^{\gamma} \nn \\
\gamma &=& |E(125)| + |E(235)| + |E(245)| + \sum_{b \in \cB_5} \left(|E_{b}(345)| + |E_{b}(145)| + |E_{b}(135)|\right) \,.
\ees
This formula will be the central tool in the bounds on pseudomanifolds we will derive below. Before moving to a second important formula following from the factorized expression for the amplitudes, we
remark that we can simply choose $K = 1$.

\

The second strategy we suggested to bound formula (\ref{amplitude_edge_reg}) also starts from a splitting of edge colours into two parts:
\bes
|\cA^{\cG}_{t}| &\leq& (\lambda \overline{\lambda})^{\frac{\cN}{2}} \int [\extd G]^{3 \cN} \left(\prod_{b \in \cB_5} 
\prod_{e \in E_{b}(125)\cup E_{b}(235) \cup E_{b}(245)}  \delta_{\langle e,b \rangle t}(\overrightarrow{\prod_{\tau \supset e}} (G^{\tau}_{e})^{\epsilon^{\tau}_{e}}) \right) \nn \\
&&\times  \left(\prod_{b \in \cB_5}  \prod_{e \in E_{b}(345)\cup E_{b}(145) \cup E_{b}(135)}  \delta_{\langle e,b \rangle t}( \overrightarrow{\prod_{\tau \supset e}}
(G^{\tau}_{e})^{\epsilon^{\tau}_{e}})\right)
\left( \prod_{\tau \in T_5} \delta(G^{\tau}_{43} G^{\tau}_{31} G^{\tau}_{14}) \delta(G^{\tau}_{32} G^{\tau}_{21} G^{\tau}_{13}) \delta(G^{\tau}_{24} G^{\tau}_{41} G^{\tau}_{12}) \right) \nn 
\ees
We can integrate the propagators with respect to variables $G^{\tau}_{23}$ and $G^{\tau}_{24}$, which are only involved in the first factor:
\bes
|\cA^{\cG}_{t}| &\leq& (\lambda \overline{\lambda})^{\frac{\cN}{2}} \int [\extd G]^{2 \cN} \left( \prod_{b \in \cB_5} \prod_{e \in E_{b}(125)} 
\delta_{\langle e,b \rangle t}(\overrightarrow{\prod_{\tau \supset e}} (G^{\tau}_{21})^{\epsilon^{\tau}_{e}}) \right) \nn \\
&&  \left( \prod_{b \in \cB_5}  
  \left( \prod_{e \in E_{b}(235)} \delta_{\langle e,b \rangle t}(\overrightarrow{\prod_{\tau \supset e}} (G^{\tau}_{21} G^{\tau}_{13})^{\epsilon^{\tau}_{e}}) \right)
  \left( \prod_{e \in E_{b}(245)} \delta_{\langle e,b \rangle t}(\overrightarrow{\prod_{\tau \supset e}} (G^{\tau}_{21} G^{\tau}_{14})^{\epsilon^{\tau}_{e}}) \right) \right)  \nn \\
&&\times   \left(\prod_{b \in \cB_5}  \prod_{e \in E_{b}(345)\cup E_{b}(145) \cup E_{b}(135)}  \delta_{\langle e,b \rangle t}( \overrightarrow{\prod_{\tau \supset e}}
(G^{\tau}_{e})^{\epsilon^{\tau}_{e}})\right)
\left( \prod_{\tau \in T_5} \delta(G^{\tau}_{43} G^{\tau}_{31} G^{\tau}_{14}) \right) 
\ees
The interesting feature of this formula is the following: the last line is the integrand of the Boulatov amplitude of the 3d coloured graph obtained from $\cG$ by deleting the lines of colour $2$! 
Since the connected components of this graph are the bubbles in $\cB_{2}$, it is tempting to factorize their 3d amplitudes.
The simplest way do so is to bound the $\delta$-functions appearing in the first two lines that involve variables from the third. We wrote the last inequality in such a way that these are exactly the
terms in the second line. 
Therefore:
\bes
|\cA^{\cG}_{t}| &\leq& (\lambda \overline{\lambda})^{\frac{\cN}{2}} \left( \prod_{b \in \cB_5} \prod_{e \in E_{b}(235) \cup E_{b}(245)} \delta_{\langle e,b \rangle t}(\one) \right) 
\int [\extd G]^{\frac{\cN}{2}} \left( \prod_{b \in \cB_5} \prod_{e \in E_{b}(125)}  \delta_{\langle e,b \rangle t}(\overrightarrow{\prod_{\tau \supset e}} (G^{\tau}_{21})^{\epsilon^{\tau}_{e}})
\right) \nn \\
&&\times   \int [\extd G]^{\frac{3 \cN}{2}} \left(\prod_{b \in \cB_5}  \prod_{e \in E_{b}(345)\cup E_{b}(145) \cup E_{b}(135)}  \delta_{\langle e,b \rangle t}( \overrightarrow{\prod_{\tau \supset e}}
(G^{\tau}_{e})^{\epsilon^{\tau}_{e}})\right)
\left( \prod_{\tau \in T_5} \delta(G^{\tau}_{43} G^{\tau}_{31} G^{\tau}_{14}) \right) 
\ees
The last line is now exactly a product of bubble 3d amplitudes: $\prod_{b \in \cB_2} \cA^{b}_{t}$. As for the integral in the first line, it is associated with the graph made of all the strands of
colour $(125)$ which, as in the previous case, we bound by:
\beq
\int [\extd G]^{\frac{\cN}{2}} \left( \prod_{b \in \cB_5} \prod_{e \in E_{b}(125)}  \delta_{\langle e,b \rangle t}(\overrightarrow{\prod_{\tau \supset e}} (G^{\tau}_{21})^{\epsilon^{\tau}_{e}})
\right)
\leq \prod_{e \in E(125)}  \delta_{|e| t}( \one) \, .
\eeq

The net result is that, for any constant Q such that:
\beq
Q > Q_0 \equiv \left(\prod_{b \in \cB_5}  \prod_{e \in E_{b}(235) \cup E_{b}(245)}  \langle e,b \rangle^{- 3} \right) \left( \prod_{e \in E(125)} |e|^{-3} \right)
\eeq
we asymptotically have:
\bes\label{bound2}
|\cA^{\cG}_{t}| &\leq& Q \,(\lambda \overline{\lambda})^{\frac{\cN}{2}} \, [\delta_{t}(\one)]^{\eta} \prod_{b \in \cB_2} |\cA^{b}_{t}| \nn \\
\eta &=& |E(125)| +\sum_{b \in \cB_5} \left(|E_{b}(235)| + |E_{b}(245)|\right) \,.
\ees
As before, we remark that $Q = 1$ is a valid choice.

\subsection{Bounding pointlike singularities}


In this section, we would like to use equation (\ref{bound1}) to give a bound on the amplitudes of simplicial complexes with pointlike singularities. As in 3d, we therefore need to introduce a
combinatorial quantity that 
determines whether a bubble is spherical or not. Unlike in 3d however, this cannot be the genus, since the Euler characteristic of a manifold in odd dimensions is 0, irrespectively of its topology. We
therefore propose 
to use the notion of degree \cite{large-N} instead:  $\omega = \sum_{J} g_J$, where the sum runs over the jackets of the 4-coloured graph dual to the triangulation of the boundary of the 4d bubble
around vertices of the simplicial complex, and $g_J$ is the genus of the jacket. Indeed, the degree is 0 only for a specific set of simplicial decompositions of the sphere (those associated to melonic
graphs \cite{critical}). As a result, in this 4d case we will prove a bound on simplicial complexes which have non-melonic bubbles, that is on a subclass of manifolds as well as on singular pseudomanifolds.
But, remarking that {\it{all}} the jackets of a non-spherical bubble have genus bigger than $1$, we will refine the bound for singular pseudomanifolds, which will in the end take the same form as in 3d.

\
Let us consider a connected and closed coloured graph $\cG$. We propose to use the bound (\ref{bound1}) to deduce a bound on the divergent degree $\gamma_{\cG}$ of $\cG$. Since the prefactor $K$ can be
chosen to be $1$, only the parameter $\gamma$ in (\ref{bound1}) matters, and $\gamma_{\cG} \leq \gamma$, that is:
\bes\label{bound1'}
\gamma_{\cG} &\leq& |E(125)| + |E(235)| + |E(245)| + \sum_{b \in \cB_5} \left(|E_{b}(345)| + |E_{b}(145)| + |E_{b}(135)|\right) \,.
\ees
Since this expression involves only quantities associated to edges and $4$-bubbles of the simplicial complex, let us show that the sum of the degrees of the same bubbles can also be expressed in terms
of similar quantities, namely we first prove:
\begin{lemma}\label{lemmaAdd}
\beq\label{sum_deg}
\sum_{b \in \cB_5} \omega(b) = 3 |\cB_{5}| +  \frac{3 \cN}{2} - \sum_{b \in \cB_5} |E_{b}|
\eeq
\end{lemma}
\begin{proof}
Indeed, by definition of the degree of a bubble $b \in \cB_{5}$, and by formula (\ref{g_jacket}), we have:
\beq		
\omega(b) = \sum_{ J= (\sigma, \sigma^{- 1})} \left( 1 + \frac{1}{2} \left( |T_{b}^{5}| - \sum_{\ell = 1}^{4} |E_{b}(\sigma(\ell) \sigma(\ell + 1))| \right) \right)\,,
\eeq
where $\sigma$ are cycles over $\{ 1, \ldots, 4\}$.\footnote{Note that we also used the fact that labelling jackets by cycles over vertex colours as in (\ref{g_jacket}), is equivalent to labelling them
by cycles over edge colours.} There are three different jackets in $b$, and each of them has support over four edge colours out of six, therefore:
\bes
\omega(b) &=& 3 +  \frac{1}{2} \left( 3 |T_{b}^{5}| - \frac{3 \times 4}{6} |E_b| \right) \nn \\
&=& 3 +  \frac{3 |T_{b}^{5}|}{2} - |E_b|\,.
\ees
Summing over the bubbles, and using
\beq
\sum_{b \in \cB_5} |T_{b}^{5}| = 2 |T^{5}| = \cN
\eeq
we finally obtain the claimed equality. \end{proof}
We therefore just need to make $\sum_{b \in \cB_5} |E_{b}|$ appear on the right-hand-side of inequality (\ref{bound1'}) To this aim, and as in the 3d case, we prove a combinatorial lemma:
\begin{lemma}\label{comb_lemma2}
 For any distinct edge colour $i$:
\beq
|\cB_{5}| + |E(i)| - \sum_{b \in \cB_{5}} |E_{b}(i)| \leq 1 \, .
\eeq
\end{lemma}
\begin{proof} 
Consider the graph $\cC_{i,5}(\cG)$ whose nodes are the elements of $\cB_5 \cup E(i)$, and links are constructed as follows: there is a link between a bubble $b \in \cB_5$ and $e \in E(i)$ if and only
if $e$ appears in the triangulation
of $b$.
\
We first show that, because $\cG$ itself is assumed to be connected, $\cC_{i,5}(\cG)$ is also connected. Let $b_i$ and $b_f$ be two elements of $\cC_{i,5}(\cG) \cap \cB_5$. The connectivity of $\cG$
ensures that there exists
a sequence of $p$ bubbles $(b_i \equiv b_0 , b_1 , \cdots , b_{p-1} , b_p \equiv b_f)$, such that: for any $k \in \llbracket 0 , p\rrbracket$, $b_k$ and $b_{k + 1}$ share a tetrahedron $\tau_k$. 
Hence: for any $k \in \llbracket 0 , p\rrbracket$, $b_k$ and $b_{k + 1}$ share an edge $e_k$. So, in $\cC_{i, 5}(\cG)$, $(b_i \equiv b_0 , e_0 , b_1 , e_1, \cdots , b_{p-1} , e_{p - 1}, b_p \equiv
b_f)$ is a connected path
from $b_i$ to $b_f$. Finally, remarking that any element of $\cC_{i, 5}(\cG) \cap E(i)$ is by definition connected to at least one element of $\cC_{i, 5}(\cG) \cap \cB_5$, we conclude that any two
nodes of $\cC_{i, 5}(\cG)$
are connected.
\
Call $L$ the number of links of $\cC_{i, 5}(\cG)$. Its number of nodes being equal to $|\cB_5| + |E(i)|$, the connectivity implies that:
\beq
|\cB_5| + |E(i)| \leq 1 + L\,.
\eeq
But for any $b \in \cB_5$, the number of lines $L(b)$ connected to $b$ in $\cC_{i, 5}(\cG)$ verifies: $L(b) \leq |E_{b}(i)|$. This is because $E_{b}(i)$ is the set of available edges of colour $i$ in
the bubble (some of which
might be identified when gluing the different bubbles together), and $L(b)$ is the true number of edges of colour $i$ in $b$, once all the identifications of edges have been taken into account. Since
by construction
\beq
L= \sum_{b \in \cB_5} L(b)\,,
\eeq 
we conclude that:
\beq
|\cB_5| + |E(i)| \leq 1 + \sum_{b \in \cB_5} L(b) \leq 1 + \sum_{b \in \cB_5} |E_{b}(i)|\,.
\eeq
\end{proof}

This, together with (\ref{bound1'}) and (\ref{sum_deg}), immediately yields the wanted result:
\beq
\gamma_{\cG} \leq 3 + \frac{3 \cN}{2}  - \sum_{b \in \cB_{5}} \omega(b) \, .
\eeq
Of course, we have a similar bound for any colour $\ell$. 

This formula is particularly simple when the coupling constant $\lambda$ is appropriately rescaled:
\beq\label{rescaling_4d}
 \lambda \rightarrow \frac{\lambda}{\delta_{t}(\one)^{3/2}}\,,
\eeq
which is also the setting of the $1/N$ expansion \cite{large-N}. Thus we have obtained the following:
\begin{proposition}
With the rescaling of the coupling constant (\ref{rescaling_4d}), the divergence degree $\gamma_{\cG}$ of a graph $\cG$ verifies, for any colour $\ell$:
\beq
\gamma_{\cG} \leq 3  - \sum_{b \in \cB_{\ell}} \omega(b) \, \label{bound}.
\eeq
\end{proposition}

Given that any pointlike singularity implies a degree greater than zero, we therefore see that, like in 3d, the more pointlike singularities of the same colour a simplicial complex has, the more its
amplitude is suppressed with respect to the leading order 
($\sim \delta_{t}(\one)^{3}$). 
We can make this point a bit more precise, making use of a lemma, a simple proof of which can be found in the review \cite{tensorReview}\footnote{This result is actually generalizable to any dimension, since a graph with a planar
jacket has a trivial \textit{regular genus}, a property that in turn characterizes spheres (see \cite{FerriGagliardi, Vince_gene} and references therein).}:
\begin{lemma}
 If a connected $4$-coloured graph $\cG$ possesses a spherical jacket, then $\cG$ is dual to a sphere.
\end{lemma}
\begin{proof}
 See Proposition 3 in \cite{tensorReview}.
\end{proof}
In particular, this implies that if a bubble $b$ is not spherical, all its jackets have genus at least $1$, and therefore: $\omega(b) \geq 3$. This allows to prove the following corollary:
\begin{corollary}
 With the rescaling of the coupling constant (\ref{rescaling_4d}), the divergence degree $\gamma_{\cG}$ of any connected vacuum graph $\cG$ verifies, for any colour $\ell$:
\beq
\gamma_{\cG} \leq 3 (1 - N^{s}_{\ell}) \,,
\eeq
where $N^{s}_{\ell}$ is the number of singular vertices of colour $\ell$. 
\end{corollary}
This result is very similar to what we have proven in 3d, showing suppression of singular pseudomanifolds with respect to the leading order ($\gamma_{\cG} = 3$), therefore consisting only of
manifolds. In particular, in both cases singular pseudomanifolds all have convergent amplitudes. 

Finally, as anticipated in the beginning of this section, we notice that equation (\ref{bound}) also constrain the amplitudes of a special class of manifolds: those which have non-melonic bubbles.
\begin{corollary}
With the rescaling of the coupling constant (\ref{rescaling_4d}), the divergence degree $\gamma_{\cG}$ of any connected vacuum graph $\cG$ verifies, for any colour $\ell$:
\beq
\gamma_{\cG} \leq 3 - N^{nm}_{\ell} \,,
\eeq
where $N^{nm}_{\ell}$ is the number of non-melonic bubbles of colour $\ell$. 
\end{corollary}

\subsection{Jacket bound}

The notion of jacket of a graph $\cG$ we used in 3d generalizes to any dimension. They are closed surfaces labelled by pairs $(\sigma , \sigma^{\inv})$ of cyclic permutations of $\{ 1 , \ldots , 5
\}$, 
in which the graph $\cG$ can be regularly embedded \cite{FerriGagliardi, Vince_gene}\footnote{In contrast with the 3d case, we construct the jackets in terms of data related to the graph $\cG$ itself,
as opposed to its dual simplicial complex. 
If the two descriptions are of course equivalent, we find more convenient to work with $\cG$ itself in 4d, since representing or even giving an intuitive picture of simplicial complexes is necessarily
more difficult than in 3d.}.
The jacket $J = (\sigma , \sigma^{\inv})$ is the closed surface constituted of all the faces of colours $(\sigma(\ell) \sigma(\ell + 1))$ in $\cG$, glued along there common links. Its genus can be
easily computed
in terms of combinatorial data associated to $\cG$:
\begin{lemma}
 The jacket $J = (\sigma , \sigma^{\inv})$ of a $5$-coloured graph $\cG$ has genus:
\beq
g_{J} = 1 + \frac{3 \cN}{4} - \frac{1}{2} \sum_{\ell = 1}^{5} |\cF(\sigma(\ell) \sigma(\ell + 1))|\,,
\eeq
where $\cN$ is the order of $\cG$, and $\cF( i j )$ is the set of faces of colour $(i j)$ (dual to the set of triangles of colour $(ij)$ in the simplicial complex).

\end{lemma}
 \begin{proof}
By definition, the Euler characteristic of $J$ is:
\beq
\chi_J = 2 - 2 g_J = |\cF_J| - |\cE_J| + |\cV_J|
\eeq
where $\cF_J$, $\cE_J$ and $\cV_J$ are the set of faces, edges and vertices of J. But, by construction:
\beq
|\cF_J| = \sum_{\ell} |\cF(\sigma(\ell) \sigma(\ell + 1))| \;, \qquad |\cE_J| = \cL \;, \qquad |\cV_J| = \cN \;,
\eeq
where $\cL$ is the number of lines in $\cG$. Since moreover $\cL = 2 \cN$, we conclude that:
\beq
2 - 2 g_J = - \frac{3 \cN}{2} + \sum_{\ell} |\cF(\sigma(\ell) \sigma(\ell + 1))| \,.
\eeq 
 \end{proof}

Jackets have been used in the Ooguri model to compute bounds on amplitudes \cite{large-N}. Interestingly, our construction allows to slightly strengthen these results, which we demonstrate now.
Discarding prefactors, formula (\ref{bound2}) immediately implies the following bound on the degree of divergence of a graph $\cG$:
\beq
\gamma_{\cG} \leq |E(125)| + \sum_{b \in \cB_5} \left(|E_{b}(235)| + |E_{b}(245)|\right) + \sum_{b \in \cB_2} \gamma_{3d}(b)\,,
\eeq 
where $\gamma_{3d}(b)$ is the degree of divergence of the $3d$ amplitudes represented by the bubble graph $b$.
Since the choice of colours is arbitrary this generalizes immediately. For any cycle $\sigma = (\ell_1  \ell_2  \ell_3  \ell_4 \ell_5)$ of $(12345)$:
\beq\label{ineq}
\gamma_{\cG} \leq |E(\ell_1 \ell_3 \ell_5)| + \sum_{b \in \cB_{\ell_1}} \left(|E_{b}(\ell_1 \ell_2 \ell_3)| + |E_{b}(\ell_1 \ell_3 \ell_4)|\right) + \sum_{b \in \cB_{\ell_3}} \gamma_{3d}(b)\,.
\eeq 

We would like to deduce from this an inequality involving the genus of the jacket $J$ associated to the cycle $\sigma$. This can be done in three steps. First remark that for any $i \neq j$, 
and $\ell$ different from both $i$ and $j$, the set of faces $\cF(i j)$ can be partitioned in terms of the faces $\cF_{b}(i j)$ of the bubble graphs $b \in \cB_{\ell}$:
\beq
\cF(ij) = \underset{b \in \cB_{\ell}}{\cup} \cF_{b}(i j)\,. 
\eeq
But since the set of faces $\cF_{b}(ij)$ is dual to the set of edges $\E_{b}(\ell ij)$ in the simplicial complex associated to $b$, we immediately obtain the following equality of cardinals:
\beq
|\cF(i j)| = \sum_{b \in \cB_\ell} |E_{b}(\ell i j)| \,.
\eeq
In particular, the sum over $\cB_{\ell_1}$ in (\ref{ineq}) is equal to $|\cF(\ell_2 \ell_3)| + |\cF(\ell_3 \ell_4)|$. Second, we can use lemma \ref{comb_lemma2} to bound $|E(\ell_1 \ell_3 \ell_5)|$:
\beq
|E(\ell_1 \ell_3 \ell_5)| \leq 1 - |\cB_{\ell_3}| + \sum_{b \in \cB_{\ell_3}} |E_{b}(\ell_1 \ell_3 \ell_5)| = 1 - |\cB_{\ell_3}| + |\cF(\ell_1 \ell_5)|\,.
\eeq
Finally, we can use a bound of the type (\ref{ineq_3d}) to bound the $\gamma_{3d}(b)$ terms:
\beq
\sum_{b \in \cB_{\ell_3}} \gamma_{3d}(b) \leq \sum_{b \in \cB_{\ell_3}} \left( 1 + |E_{b}(\ell_3 \ell_1 \ell_2)| + |E_{b}(\ell_3 \ell_4 \ell_5)| \right) = |\cB_{\ell_3}| + |\cF(\ell_1 \ell_2)| +
|\cF(\ell_4 \ell_5)|\,.
\eeq
All in all we obtain:
\beq
\gamma_{\cG} \leq 1 + \sum_{\ell = 1}^{5} |\cF( \sigma(\ell) \sigma(\ell + 1))| = 3 + \frac{3 \cN}{2} - 2 g_J
\eeq
This is our final result, which as in 3d is particularly nice once the coupling constant has been appropriately rescaled.
\begin{proposition}\label{jacket_4d}
With the rescaling of the coupling constant (\ref{rescaling_4d}), the divergence degree $\gamma_{\cG}$ of a connected vacuum graph $\cG$ verifies, for any jacket $g_J$:
\beq
\gamma_{\cG} \leq 3  - 2 g_J\,.
\eeq
\end{proposition}

In particular, we notice that the best (strongest) bound is obtained considering the {\it maximum} of the genera of all the 12 jackets of the graph $\cG$.
Also, we see that all graphs with a jacket of genus greater than $2$ have convergent amplitudes, just like singular pseudomanifolds. 

\

This has to be compared with the known jacket bound, which is weaker and is a direct corollary of the previous result:
\begin{corollary}
 With the rescaling of the coupling constant (\ref{rescaling_4d}), the divergence degree $\gamma_{\cG}$ of a connected vacuum graph $\cG$ verifies:
\beq
\gamma_{\cG} \leq 3  - \frac{1}{6} \omega(\cG)\,.
\eeq
\end{corollary}

\begin{proof}
 The degree of $\cG$ is defined by:
\beq
\omega(\cG) = \sum_{J} g_J\,,
\eeq
so that averaging the previous bound over the $12$ jackets of $\cG$ immediately gives the standard jacket bound. 
\end{proof}

\section{Conclusion}
We have used a vertex reformulation of the Boulatov model in 3d and an edge reformulation of the Ooguri model in 4d to obtain new scaling bounds for their Feynman amplitudes. The standard formulations
in terms of variables associated to edges (Boulatov) and triangles (Ooguri) have the advantage of bringing the same amplitudes in a geometrically clear simplicial path integral form (in the Lie
algebra representation) or of highlighting the topological nature of the model, weighting only flat discrete connections (in group representation). The new formulations, on the other hand, seem more
natural from a field theory point of view, in light of the fact that the symmetry group of the theory (diffeomorphism/translation symmetry) acts on such variables, and, as we have shown in this paper, allows
a more direct and powerful analysis of the scaling of the amplitudes when the cut-off is removed. In particular, we have obtained, in both 3d and 4d topological models: a) a bound focused on the
contribution of {\it bubbles} (3-cells or 4-cells, in 3d and 4d respectively) around vertices of the simplicial complex dual to the GFT diagrams, where the presence of conical singularities
distinguishes arbitrary pseudomanifolds from manifolds, showing that the singular complexes are generically suppressed in perturbation theory (that is, at any order or for any number of simplices); b)
a bound depending on the genera of the so-called {\it jackets} (embedded surfaces) of the same complex, stronger than the known ones that allowed to prove the dominance of spherical manifolds in the
limit of large values of the cut-off. In the process, we have obtained several further insights into the combinatorial structure of the diagrams themselves and the properties of the amplitudes.

\

Let us comment on a few interesting aspects of the results obtained, and of the procedure used to obtain them. First, as we anticipated, the entire construction relies on the gauge invariance/closure
condition imposed in the initial GFT field. Therefore it could be considered, at least until similar results are reproduced by other means, a characteristic feature of GFT models of the type we
considered (and in more direct relation with spin foam and simplicial gravity approaches), not easily extendable to the simpler tensor models. Second, at the same time the geometric interpretation and
content of the data entering the GFT amplitudes and states has not been really used in a crucial way in our analysis, which was essentially based on combinatorial aspects, thus more in line with the
usual analysis of tensor models. The full geometric content will become more relevant, we expect, when considering open GFT Feynman diagrams and non-trivial boundary states. From both a physical and
mathematical point of view, we expect the type of boundary states used to affect sensibly the sum over complexes, that is the dynamics of the theory. Of course, this will be even more true in
non-topological, 4d gravity models.

\

On top of settling the important issue of the relative weight of regular versus singular configurations in the GFT sum over complexes, at least for this class of models, the results obtained in this
work, and in particular the improved jacket bounds, will be relevant first of all for future studies of renormalizability and critical behaviour of the same models. This is indeed the next item on the
agenda.

We expect them to be directly relevant also for closely related models. These include, for example, the dynamical version of topological models, that is their modification including group laplacian
operators in the kinetic term. The additional propagators would modify the power counting, but on the one hand they just bring additional decays with the cut-off, on the other hand they may also be an
alternative way to obtain the rescaling of the coupling constant that we imposed in deriving uniform bounds. Thus we expect our scaling bounds  and in general our analysis to be still relevant. Still,
only a proper study can confirm this expectation. More important, our result will be relevant for 4d gravity models. In fact, all current spin foam and GFT models for 4d gravity are constructed
starting from the topological Ooguri model by adding suitable restrictions on the group, Lie algebra or representation data appearing in the GFT action and summed over in the Feynman amplitudes. As a
result, they share many properties of topological models and we expect our analysis, if not our results, to be easily extended to them. As we had pointed out already, the fact that our proofs did not
rely on combinatorial/topological invariance properties (e.g. invariance of the amplitudes under Pachner or dipole moves) should make the extension  to gravity models less difficult.

The edge representation we develop for the topological 4d GFT model has a potential interest for the 4d gravity models from a purely geometric perspective as well. In fact, as mentioned above, such
gravity models are constructed by constraining the variables associated to triangles of the simplicial complex and that are interpreted as the discretization of the B-field in topological BF theory,
an interpretation that is confirmed explicitly by the simplicial path integral form of the GFT amplitudes in the Lie algebra representation. The sought-for effect of the constraints on discrete
B-variables is to enforce geometricity of the configurations summed over in the resulting simplicial path integral, or, in other words, to allow the reconstruction from the same discrete B-variables
of a discrete tetrad field encoding the metric information of the simplicial complex. The resulting discrete tetrad would be associated indeed to the edges of the simplicial complex, as edge vectors.
Therefore, we expect that the edge reformulation of the constrained topological GFT models, following our construction, will lead more naturally to a re-formulation of them in which metric information
will be more transparent and easier to analyse. In the same 4d gravity models we may expect a further transformation of variables to be available, which would express the GFT field and amplitudes in
terms of variables associated to the vertices of the simplicial complex. Beside purely geometric considerations, this expectation is motivated by the fact that the constraints reducing topological BF
theory to gravity in a simplicial setting also break the translation invariance of the theory (acting on edges of the complex), leaving as natural transformations (although not necessarily symmetries
\cite{diffeosRegge}) only vertex translations, as a simplicial counterpart of diffeomorphism invariance of the continuum. We leave this for future work. 

\section*{Acknowledgements}
This work is partially supported by a Sofja Kovalevskaja Award by the A. von Humboldt Stiftung, which is gratefully acknowledged. It is a pleasure to thank Aristide Baratin and James Ryan for fruitful
discussions.

\end{document}